\documentclass{article}%
\usepackage{amsmath}
\usepackage{amsfonts}
\usepackage{amssymb}
\usepackage{graphicx}
\usepackage{dsfont}%
\providecommand{\U}[1]{\protect\rule{.1in}{.1in}}
\newtheorem{theorem}{Theorem}

\newtheorem{definition}[theorem]{Definition}

\newtheorem{lemma}[theorem]{Lemma}

\newtheorem{proposition}[theorem]{Proposition}

\newenvironment{proof}[1][Proof]{\noindent\textbf{#1.} }{\ \rule{0.5em}{0.5em}}
\begin{document}

\title{A new class of problems in the calculus of variations}
\author{Ivar Ekeland$^{1}$\thanks{E-mail: ekeland@math.ubc.ca},\quad Yiming Long$^{2}%
$\thanks{Partially supported by NNSF, MCME, RFDP, LPMC of MOE of China, and
Nankai University. E-mail: longym@nankai.edu.cn},\quad Qinglong Zhou$^{2}%
$\thanks{Partially supported by the Chern Institute of Mathematics, Nankai
University, CEREMADE of Universit\'{e} Paris Dauphine, and the IHES. E-mail:
zhou.qinglong.1985@gmail.com}\\
\\$^{1}$ CEREMADE and Institut de Finance\\Universit\'{e} de Paris-IX, Dauphine, Paris, France\\$^{2}$ Chern Institute of Mathematics and LPMC\\Nankai University, Tianjin 300071, China\\}
\maketitle


\begin{center}
{\large Dedicated to Professor Alain Chenciner on his 70th birthday}
\end{center}

\section{ Introduction}

In economic theory, and in optimal control, it has been customary to discount
future gains at a constant rate $\delta>0$. If an individual with utility
function $u\left(  c\right)  $ has the choice between several streams of
consumption $c\left(  t\right)  $, $0\leq t$, he or she will choose the one
which maximises the present value, given by:%
\begin{equation}
\int_{0}^{\infty}u\left(  c\left(  t\right)  \right)  e^{-\delta t}dt
\label{a1}%
\end{equation}

That future gains should be discounted is well grounded in fact. On the one
hand, humans prefer to enjoy goods sooner than later (and to suffer bads later
than sooner), as every child-rearing parent knows. On the other hand, it is
also a reflection of our own mortality: 10 years from now, I may simply no
longer be around to enjoy whatever I have been promised. These are two good
reasons why people are willing to pay a little bit extra to hasten the
delivery date, or will require compensation for postponement, which is the
essence of discounting.

On the other hand, there is no reason why the discount rate should be
constant, i.e. why the discount factor should be an exponential $e^{-\delta
t}$. The practice probably arises from the compound interest formula
$\lim_{\varepsilon\rightarrow0}\left(  1-\varepsilon\delta\right)
^{t/\varepsilon}=e^{-\delta t}$, when a constant interest rate $\delta$ is
assumed, but even in finance, interest rates vary with the horizon: long-term
rates can be widely different from short-term ones. As for economics, there is
by now a huge amount of evidence that individuals use higher discount rates
for the near future than for the long-term (see \cite{Frederick} for a review
up to 2002).\ There is also an aggregation problem: in a society where
individuals use constant (but different) discount rates, the collective
discount rate may be non-constant (see \cite{EkL}). So the present value
formula (\ref{a1}) should be replaced by the more general one:%
\begin{equation}
\int_{0}^{\infty}u\left(  c\left(  t\right)  \right)  h\left(  t\right)  dt
\label{a2}%
\end{equation}
where $h\ $\ is a decreasing function, with $h\left(  0\right)  =1$.

But then a new problem arises, which is now well recognized in economic
theory, but to our knowledge has not yet received the attention it deserves in
control theory. It is the problem of \emph{ time-inconsistency}, which runs as
follows. Suppose the decision-maker has the choice between two streams of
consumption $c_{1}\left(  t\right)  $ and $c_{2}\left(  t\right)  $, starting
at time $T>0$. At time $t=0$, he or she finds $c_{1}\left(  t\right)  $ yields
the highest present value:%

\begin{equation}
\int_{T}^{\infty}u\left(  c_{1}\left(  t\right)  \right)  h\left(  t\right)
dt>\int_{T}^{\infty}u\left(  c_{2}\left(  t\right)  \right)  h\left(
t\right)  dt. \label{a3}%
\end{equation}
He or she then chooses $c_{1}\left(  t\right)  $. When time $T$ is reached,
the present values are now:%
\begin{equation}
\int_{T}^{\infty}u\left(  c_{1}\left(  t\right)  \right)  h\left(  t-T\right)
dt\text{ \ and}\ \ \int_{T}^{\infty}u\left(  c_{2}\left(  t\right)  \right)
h\left(  t-T\right)  dt \label{a4}%
\end{equation}

If $h\left(  t\right)  =e^{-\delta t}$, then the ordering found at time $t=0$
will persist at time $t=T$. Indeed:%
\[
\int_{T}^{\infty}u\left(  c\left(  t\right)  \right)  e^{-\delta\left(
t-T\right)  }dt=e^{\delta T}\int_{T}^{\infty}u\left(  c\left(  t\right)
\right)  e^{-\delta t}dt
\]
so that the two terms in (\ref{a4}) are proportional to the two terms in
(\ref{a3}). However, this is a peculiarity of the exponential function, and it
is not to be expected with more general discount rates. The decision-maker
then faces a basic rationality problem: what should he or she do ? To be more
specific, assume the state $k\left(  t\right)  $ is related to the control
$c\left(  t\right)  $ by the dynamics:%
\begin{align}
\frac{dk}{dt}  &  =f\left(  k\right)  -c\left(  t\right)  ,\ \ k\left(
0\right)  =k_{0}\label{a7}\\
c\left(  t\right)   &  \geq0,\ \ k\left(  t\right)  \geq0 \label{a8}%
\end{align}
and the decision-maker is interested in maximising $\left(  \ref{a2}\right)
$. How should he or she behave ?

In the exponential case, when $h\left(  t\right)  =e^{-\delta t}$, the answer
is to pick the optimal solution: if it is optimal at time $t=0$, it will still
be optimal at all times $T>0$ (this, by the way, is the content of the dynamic
programming principle). But in the non-exponential case, the notion of
optimality changes with time: each observer, from time $t=0$ on, has his or
her own optimal solution. No one agrees on what the optimal solution is, so
optimality no longer provides an answer to the decision-making process, and
one must look for other concepts to describe rational behaviour.

A clear requirement for rationality is that any strategy put forward be
implementable. Suppose a Markov strategy $c=\sigma\left(  k\right)  $ is put
forward at time $t=0$. If it is to be followed at all later times $t>0$, then
it must be the case that the decision-maker at that time finds no incentive to
deviate. More precisely, if he/she assumes that at all later times the
strategy (closed-loop feedback) $c=\sigma\left(  k\right)  $ will be applied,
then he/she should find it in his/her interest to apply $\sigma$ as well. In
other words, $\sigma$ should be a subgame-perfect Nash equilibrium of the
leader-follower game played by the successive decision-makers.\ This idea has
been introduced by Phelps (\cite{Phelps}, \cite{Phelps1}) in models with
discrete time (see \cite{Krusell} and \cite{Laibson} for further
developments), and adapted by Karp (\cite{Kar2},\ \cite{KaL1}), and by Ekeland
and Lazrak (\cite{EkL1}, \cite{EkL2}, \cite{EkL}, \cite{EKS}) to the case of
continuous time.

In this paper, we will follow the approach by Ekeland and Lazrak. It consists
of introducing a value function $V\left(  k\right)  $, which is very similar
to the value function in optimal control, and of showing that it satisfies a
functional-differential equation which is reminescent of the
Hamilton-Jacobi-Bellman (HJB)\ of optimal control. Conversely, any solution of
that equation with suitable boundary conditions will give us an equilibrium strategy.

In the work by Ekeland and Lazrak, this approach was applied to (\ref{a2}),
with $h\left(  t\right)  =\alpha\exp\left(  -r_{1}t\right)  +\left(
1-\alpha\right)  \exp\left(  -r_{2}t\right)  $, and it was showed that the
corresponding problem had a continuum of equilibrium strategies.\ In the
present paper, in view of applications to economics, and of the mathematical
interest, we aim to extend the analysis to the more general case:%
\begin{equation}
\left(  1-\alpha\right)  \int_{0}^{\infty}e^{-r_{1}t}u\left(  c\left(
t\right)  ,k\left(  t\right)  \right)  dt+\alpha\int_{0}^{\infty}e^{-r_{2}%
t}U\left(  c\left(  t\right)  ,k\left(  t\right)  \right)  dt \label{a5}%
\end{equation}

As a by-product of our analysis, we will treat the problem:%
\begin{equation}
\left(  1-\alpha\right)  \delta\int_{0}^{\infty}e^{-\delta t}u\left(  c\left(
t\right)  ,k\left(  t\right)  \right)  dt+\alpha\lim_{t\rightarrow\infty
}U\left(  k\left(  t\right)  ,c\left(  t\right)  \right)  \label{a6}%
\end{equation}
which was introduced by Chichilnisky (see \cite{Chi2}, \cite{Chi1}) to model
sustainable development. Note that, if $u\left(  c,k\right)  \geq0$ and:%
\[
\sup\left\{  U\left(  c,k\right)  \ |\ c>0,\ k>0,\ c=f\left(  k\right)
\right\}  =\infty
\]
then maximising (\ref{a6}) under the dynamics (\ref{a7}), (\ref{a8}) leads to
the value $+\infty$, so that optimisation is clearly not an answer to the
problem. Instead, we find equilibrium strategies. To our knowledge, this is an
entirely new result. We show that there is a continuum of such
strategies.\ More precisely, there is a continuum of points $k_{\infty}$ which
can be realized as the long-term level of capital by an equilibrium strategy.
This support, however, is one-sided, that is, $k_{\infty}$ can be reached only
from the initial level of capitals $k_{0}$ lying on its left (or from its
right). To our knowledge, this is the first time such strategies have been identified.

The structure of the paper is as follows. In the next section, we consider the
problem of maximising%
\begin{equation}
\int_{0}^{\infty}e^{-\delta t}u\left(  c\left(  t\right)  ,k\left(  t\right)
\right)  dt \label{a9}%
\end{equation}
under the dynamics (\ref{a7}), (\ref{a8}), and we show that it has a solution.
On the way, we introduce the corresponding HJB equation, and we show that it
has a $C^{2}$ solution. Next, we define equilibrium strategies. With each such
strategy we associate a value function $V\left(  k\right)  $, and we show that
it satisfies an integro-differential equation which generalizes the HJB
equation, and we prove a verification theorem: any solution of this equation
with suitable boundary conditions gives an equilibrium strategy. We show that
the trajectories satisfy an integro-differential equation which generalizes
the classical Euler-Lagrange equations, and we connect the Ekeland-Lazrak
approach with the Karp approach.

In section 4, we apply the theory to problem (\ref{a5}), and show that it has
a continuum of equilibrium strategies, thereby extending the results of
\cite{EkL}. It should be noted that the equations for $V\left(  k\right)  $
are given in implicit form, that is, they cannot be solved with respect to
$V^{\prime}\left(  k\right)  $, so that finding a $C^{2}$ solution requires
special techniques (first a blow-up, and then the central manifold theorem).

We then consider the criterion:%
\begin{equation}
\left(  1-\alpha\right)  \delta\int_{0}^{\infty}e^{-\delta t}u\left(  c\left(
t\right)  ,k\left(  t\right)  \right)  dt+\alpha r\int_{0}^{\infty}%
e^{-rt}U\left(  c\left(  t\right)  ,k\left(  t\right)  \right)  dt \label{a10}%
\end{equation}
which belongs to the class (\ref{a5}) and we let $r\rightarrow0$. In the
limit, we get equilibrium strategies for the Chichilnisky problem (\ref{a6}),
which we describe explicitly.

\section{The Ramsey problem\label{s1}}

This is the classical model for economic growth, originating with the seminal
paper of Ramsey \cite{Ram1} in 1928, and developed by Cass \cite{Cass},
Koopmans \cite{Koop} and many others (see \cite{Barro} for a modern
exposition). We are given a point $k_{0}>0$ and a concave continuous function
$f$ on $[0,\ \infty)$, which is $C^{\infty}$ on $]0,\ \infty\lbrack$ and
satisfies the Inada conditions:%
\begin{equation}
\label{Inada}\lim_{x\rightarrow0}f^{\prime}\left(  x\right)  =+\infty
,\ \ \lim_{x\rightarrow\infty}f^{\prime}\left(  x\right)  \leq0.
\end{equation}

\begin{definition}
A capital-consumption path $\left(  c\left(  t\right)  ,k\left(  t\right)
\right)  ,\ t\geq0$, is \emph{admissible} if:
\begin{align}
k\left(  t\right)   &  >0\text{ and }c\left(  t\right)  >0\text{ for all
}t,\label{23}\\
\frac{dk}{dt}  &  =f\left(  k\right)  -c,\ \ k\left(  0\right)  =k_{0},
\label{24}%
\end{align}

\end{definition}

The set of all admissible paths starting from $k_{0}$ (i.e, such that
$k\left(  0\right)  =k_{0}$) will be denoted by $\mathcal{A}\left(
k_{0}\right)  $.

We are given a number $\delta>0$ and another concave, increasing function $u$
on $]0,\ \infty)$, which is $C^{\infty}$ on the interior, with $u^{\prime
\prime}\left(  c\right)  >0$ everywhere. We introduce the following criterion
on $\mathcal{A}\left(  k_{0}\right)  $:
\begin{equation}
I\left(  c,k\right)  =\int_{0}^{\infty}u(c(t),k(t))e^{-\delta t}dt, \label{18}%
\end{equation}
and we consider the optimization problem:%
\begin{equation}
\sup\left\{  I\left(  c,k\right)  \ |\ \left(  c,k\right)  \in\mathcal{A}%
\left(  k_{0}\right)  \right\}  . \label{19}%
\end{equation}

The Euler-Lagrange equation is given by:%
\begin{equation}
u_{11}^{\prime\prime}\frac{dc}{dt}=\left(  \delta-f^{\prime}\left(  k\right)
\right)  u_{1}^{\prime}-u_{2}^{\prime}-\left(  f\left(  k\right)  -c\right)
u_{12}^{\prime\prime}. \label{21}%
\end{equation}

Equation (\ref{21}), together with equation (\ref{24}), constitute a system of
two first-order ODEs for the unknown functions $\left(  c\left(  t\right)
,k\left(  t\right)  \right)  $. In the particular case when $u=u\left(
c\right)  $ does not depend on $k$, the last equation simplifies to:%
\[
u_{11}^{\prime\prime}\frac{dc}{dt}=\left(  \delta-f^{\prime}\left(  k\right)
\right)  u_{1}^{\prime}%
\]
and the (\ref{21}), (\ref{24}) gives rise to a well-known phase diagram, with
a hyperbolic stationary point $\left(  c_{\infty},k_{\infty}\right)  $
characterized by $f^{\prime}\left(  k_{\infty}\right)  =\delta$ and $f\left(
k_{\infty}\right)  =c_{\infty}$. The optimal solution of the Ramsey problem in
that case then is the solution of (\ref{21}), (\ref{24}) which converges to
$\left(  c_{\infty},k_{\infty}\right)  $ (see \cite{Barro} for instance).

In the general case where $u\left(  c,k\right)  $ depends on $k$, the
situation is not as simple, and to our knowledge has not been investigated.
Stationary points $\left(  c_{\infty},k_{\infty}\right)  $ of the dynamics (if
any) are given by:%

\begin{align}
c_{\infty}-f\left(  k_{\infty}\right)   &  =0\label{50}\\
\left(  \delta-f^{\prime}\left(  k_{\infty}\right)  \right)  u_{1}^{\prime
}\left(  f\left(  k_{\infty}\right)  ,k_{\infty}\right)  -u_{2}^{\prime
}\left(  f\left(  k_{\infty}\right)  ,k_{\infty}\right)   &  =0. \label{51}%
\end{align}

To prove the existence of an optimal strategy, we do not use the Euler
equation. We use the Hamilton-Jacobi-Bellman (HJB) equation instead. Introduce
the optimal value as a function of the initial point:%
\[
V\left(  k_{0}\right)  :=\sup\left\{  I\left(  c,k\right)  \ |\ \left(
c,k\right)  \ \in\mathcal{A}\left(  k_{0}\right)  \right\}
\]

If there is an optimal solution $\left(  c\left(  t\right)  ,k\left(
t\right)  \right)  $, and it converges to $\left(  c_{\infty},k_{\infty
}\right)  $ when $t\rightarrow\infty$, then, substituting in (\ref{18}), we
must have:%
\begin{equation}
V\left(  k_{\infty}\right)  =\int_{0}^{\infty}e^{-\delta t}u\left(  c_{\infty
},k_{\infty}\right)  dt=\frac{1}{\delta}u\left(  c_{\infty},k_{\infty}\right)
=\frac{1}{\delta}u\left(  f\left(  k_{\infty}\right)  ,k_{\infty}\right)
\label{20a}%
\end{equation}

\begin{theorem}
If $V\left(  k\right)  $ is $C^{1}$, it satisfies the HJB\ equation,
namely:%
\begin{equation}
\delta V(k)=\max_{c}\{u(c,k)+(f(k)-c){V}^{\prime}(k)\}. \label{20}%
\end{equation}

Conversely, suppose the HJB\ equation has a $C^{2}$ solution satisfying
(\ref{20a}) for some $\left(  c_{\infty},k_{\infty}\right)  $, and define a
strategy $c=\sigma\left(  k\right)  $ by:%
\begin{equation}
u_{1}^{\prime}\left(  \sigma\left(  k\right)  ,k\right)  =V^{\prime}\left(
k\right)  \label{20b}%
\end{equation}

Suppose moreover that the solution of:%
\begin{equation}
\frac{dk}{dt}=f\left(  k\right)  -\sigma\left(  k\right)  ,\ \ k\left(
0\right)  =k_{0}, \label{20c}%
\end{equation}
converges to $k_{\infty}$ for all initial points $k_{0}$. Then $\sigma\left(
k\right)  $ is an optimal solution of the generalized Ramsey problem (\ref{19})
\end{theorem}

This is the so-called verification theorem, which is classical (see
\cite{Bardi}). We need $V$ to be $C^{2}$, so that $\sigma$ is defined. Since
$u_{1}^{\prime\prime}>0$, we can use the implicit function theorem on equation
(\ref{20b}) to define $\sigma$.\ If $V$ is $C^{2}$, then $\sigma$ is $C^{1}$,
and the initial-value problem (\ref{20c}) has a unique solution. Note also
that everything is local: the functions $V$ and $\sigma$ are defined in some
neighbourhood of $k_{\infty}$ only, and the initial value $k_{0}$ is assumed
to belong to that neighbourhood of $k_{\infty}$.

So, to prove the (local) existence of an optimal strategy $\sigma\left(
k\right)  $, we have to prove that the HJB equation has a $C^{2}$ solution $V$
with $V\left(  k_{\infty}\right)  =\delta^{-1}u\left(  f\left(  k_{\infty
}\right)  ,k_{\infty}\right)  $, and that the corresponding path $k\left(
t\right)  $ converges to $k_{\infty}$. Then the right-hand side of (\ref{20c})
converges to $0$, so that $c\left(  t\right)  =\sigma\left(  k\left(
t\right)  \right)  $ converges to $\sigma\left(  k_{\infty}\right)  =f\left(
k_{\infty}\right)  $. This is the content of the following two results

\begin{theorem}
\label{thm.1} Suppose there is some $k_{\infty}>0$ satisfying (\ref{51}) and:%
\begin{equation}
u_{1}^{\prime}f^{\prime\prime}+u_{22}^{\prime\prime}+f^{\prime}\left(
u_{12}^{\prime\prime}-\frac{u_{2}^{\prime}}{u_{1}^{\prime}}u_{11}%
^{\prime\prime}\right)  -\frac{u_{2}^{\prime}}{u_{1}^{\prime}}u_{12}%
^{\prime\prime}<0 \label{27}%
\end{equation}
(all values to be taken at $k_{\infty}$ and $c_{\infty}=f\left(  k_{\infty
}\right)  $). Then there is an optimal strategy $c=\sigma\left(  k\right)  $
converging to $k_{\infty}$
\end{theorem}

If $u=u\left(  c\right)  $ does not depend on $k$, then $u_{2}^{\prime}=0$.
Equation (\ref{51}) becomes $f^{\prime}\left(  k_{\infty}\right)  =\delta$,
which defines $k_{\infty}$ uniquely because of the Inada conditions
(\ref{Inada}) on $f$, and condition (\ref{27}) becomes $u_{1}^{\prime}\left(
c_{\infty}\right)  f^{\prime\prime}\left(  k_{\infty}\right)  <0$, which is
satisfied automatically. So there is an optimal strategy in that case, and one
can even show that it is globally defined. For $u$ depending on $k$, however,
the situation is different.

\begin{proof}
We follow the method of \cite{Ek}. By the inverse function theorem, the
equation $u_{1}^{\prime}(c,k))=p$ defines $c$ as a $C^{1}$ function of $p$ and
$k$:
\begin{equation}
c=\varphi(p,k)\text{, }u_{1}^{\prime}(\varphi(p,k),k))=p \label{c.of.k}%
\end{equation}

We rewrite (\ref{20}) as a Pfaff system:
\begin{align}
dV  &  =pdk,\label{paff1}\\
p(f(k)-\varphi(p,k))+u(\varphi(p,k),k)  &  =\delta V, \label{paff2}%
\end{align}
and we seek a $C^{2}$ solution $V$ satisfying:
\begin{equation}
V(k_{\infty})=\frac{1}{\delta}u(f(k_{\infty}),k_{\infty}).
\label{paff.initial}%
\end{equation}

Differentiating (\ref{paff2}) leads to:
\begin{equation}
\delta dV=(f(k)-\varphi(p,k))dp+pf^{\prime}(k)dk+u_{2}^{\prime}(\varphi
(p,k),k)dk \label{delta.dv}%
\end{equation}

Plugging (\ref{delta.dv}) into (\ref{paff1}), we get:
\begin{equation}
(f(k)-\varphi(p,k))dp=(\delta{p}-pf^{\prime}(k)-u_{2}^{\prime}(\varphi
(p,k),k))dk. \label{system.tmp}%
\end{equation}

Introducing an auxiliary variable $t$, we rewrite this as a system of two ODES
for two functions $p\left(  t\right)  $ and $k\left(  t\right)  $:
\begin{equation}%
\begin{array}
[c]{c}%
\frac{dk}{dt}=f(k)-\varphi(p,k)\\
\frac{dp}{dt}=\delta{p}-pf^{\prime}(k)-u_{2}^{\prime}(\varphi(p,k),k)
\end{array}
\label{ODE}%
\end{equation}
with the initial condition
\[
(k(0),p(0))=(k_{\infty},p_{\infty})
\]

By (\ref{c.of.k}), we must have
\[
p_{\infty}=u_{1}^{\prime}(\varphi(p_{\infty},k_{\infty}),k_{\infty}%
)=u_{1}^{\prime}(c_{\infty},k_{\infty})=u_{1}^{\prime}(f\left(  k_{\infty
}\right)  ,k_{\infty})
\]

Differentiating (\ref{c.of.k}) with respect to $p$ and $k$ respectively, we
derive the formulas:
\begin{align}
\varphi_{1}^{\prime}(p_{\infty},k_{\infty})  &  =\frac{1}{u_{11}^{\prime
\prime}(\varphi(p_{\infty},k_{\infty}),k_{\infty})}=\frac{1}{u_{11}%
^{\prime\prime}(f(k_{\infty}),k_{\infty})},\label{i1'.infty}\\
\varphi_{2}^{\prime}(p_{\infty},k_{\infty})  &  =-\frac{u_{11}^{\prime\prime
}(\varphi(p_{\infty},k_{\infty}),k_{\infty})}{u_{11}^{\prime\prime}%
(\varphi(p_{\infty},k_{\infty}),k_{\infty})}=-\frac{u_{12}^{\prime\prime
}(f(k_{\infty}),k_{\infty})}{u_{11}^{\prime\prime}(f(k_{\infty}),k_{\infty})},
\label{i2'.infty}%
\end{align}

We can now linearizd (\ref{ODE}) at the point $\left(  p_{\infty},k_{\infty
}\right)  $.\ We get:%
\begin{equation}
\frac{d}{dt}\left(
\begin{array}
[c]{c}%
k-k_{\infty}\\
p-p_{\infty}%
\end{array}
\right)  =A_{\infty}\left(
\begin{array}
[c]{c}%
k-k_{\infty}\\
p-p_{\infty}%
\end{array}
\right)  . \nonumber\label{linearized.system}%
\end{equation}
where the constant matrix $A_{\infty}$ is given by:
\[
A_{\infty}:=\left(
\begin{array}
[c]{cc}%
f^{\prime}+\frac{u_{12}^{\prime\prime}}{u_{11}^{\prime\prime}}, & -\frac
{1}{u_{11}^{\prime\prime}}\\
-u_{1}^{\prime}f^{\prime\prime}-u_{22}^{\prime\prime}+\frac{(u_{12}%
^{\prime\prime})^{2}}{u_{11}^{\prime\prime}}, & \delta-f^{\prime}-\frac
{u_{12}^{\prime\prime}}{u_{11}^{\prime\prime}}%
\end{array}
\right)  .
\]
and all the values are to be taken at $\left(  k_{\infty},p_{\infty}\right)
$. The characteristic polynomial is:%
\[
{\lambda}^{2}-\delta\lambda-[\frac{u_{1\infty}^{\prime}f_{\infty}%
^{\prime\prime}+u_{22\infty}^{\prime\prime}}{u_{11\infty}^{\prime\prime}%
}+f_{\infty}^{\prime}(\frac{u_{12\infty}^{\prime\prime}}{u_{11\infty}%
^{\prime\prime}}-\frac{u_{2\infty}^{\prime}}{u_{1\infty}^{\prime}}%
)-\frac{u_{2\infty}^{\prime}}{u_{1\infty}^{\prime}}\frac{u_{12\infty}%
^{\prime\prime}}{u_{11\infty}^{\prime\prime}}]=0
\]
where we used $f_{\infty}^{\prime}+\frac{u_{2\infty}^{\prime}}{u_{1\infty
}^{\prime}}=\delta$ by (\ref{50}). Because of assumption (\ref{27}), it has
two real roots with different signs, $\lambda_{+}>0$ and $\lambda_{-}<0$. Thus
$(k_{\infty},p_{\infty})$ is a hyperbolic fixed point of (\ref{ODE}), with a
stable $C^{\infty}$-manifold $\mathcal{S}$ which corresponds to $\lambda_{-}$
and an unstable $C^{\infty}$-manifold $\mathcal{U}$ which corresponds to
$\lambda_{+}$. Choose a smooth parametrization $({k}_{s}(x),{p}_{s}(x))$ for
the curve $\mathcal{S}$. The tangent at the fixed point is:
\[
\frac{d{p}_{s}}{d{k}_{s}}(k_{\infty})=u_{11}^{\prime\prime}(f_{\infty}%
^{\prime}-\lambda_{-})+u_{12}^{\prime\prime}%
\]
and plugging $k=k_{s}(x),\ p=p_{s}(x)$ into equation (\ref{paff2}), we get
\[
V_{s}=\frac{p_{s}(x)(f(k_{s}(x))-\varphi(p_{s}(x),k_{s}(x)))+u(\varphi
(p_{s}(x),k_{s}(x)),k_{s}(x))}{\delta}.
\]

Moreover, differentiating (\ref{paff2}), and using (\ref{27}) again, we find:
\[
\frac{d{V}_{s}}{d{k}_{s}}(k_{\infty})=u_{1}^{\prime}(f(k_{\infty}),k_{\infty
})>0,
\]

It follows that the curve in parametric form $x\rightarrow\left(  k\left(
x\right)  ,V\left(  x\right)  \right)  $ is in fact the graph of a function
$V\left(  k\right)  $ which solves HJB and satisfies (\ref{paff.initial}). By
(\ref{paff1}), we have
\begin{equation}
\frac{d^{2}{V}_{s}}{d{k_{s}}^{2}}(k_{\infty})=\frac{d{p}_{s}}{d{k_{s}}%
}(k_{\infty})=u_{11\infty}^{\prime\prime}(f_{\infty}^{\prime}-\lambda
_{-})+u_{12\infty}^{\prime\prime}. \label{v''}%
\end{equation}
Thus $V_{s}(k)$ is $C^{2}$ in $k$ at $k_{\infty}$, and the $C^{2}$ property at
other points near $k_{\infty}$ follows from (\ref{ODE}).

It remains to show that the strategy $\sigma$ defined by $u_{1}^{\prime
}(\sigma\left(  k\right)  ,k)=V^{\prime}\left(  k\right)  $ converges to
$\left(  f\left(  k_{\infty}\right)  ,k_{\infty}\right)  $. We rewrite
$\sigma$ as:%
\[
\sigma(k)=\varphi(V^{\prime}(k),k).
\]

Linearizing the equation
\[
\frac{dk}{dt}=f(k)-c=f(k)-\varphi(V^{\prime}(k),k)
\]
gives
\begin{equation}
\frac{d(k-k_{\infty})}{dt}=\lambda_{-}(k-k_{\infty})\nonumber
\end{equation}
and this concludes the proof.
\end{proof}

Let us show how to deduce the Euler equation (\ref{21}) from the HJB equation
(\ref{20}) and the optimal strategy (\ref{20b}). Setting $c\left(  t\right)
:=\sigma\left(  k\left(  t\right)  \right)  $, differentiating (\ref{20}) with
respect to $k$, and applying the envelope theorem, we get:%
\[
\delta V^{\prime}(k)=u_{2}^{\prime}\left(  c,k\right)  +f^{\prime}\left(
k\right)  V^{\prime}\left(  k\right)  +\left(  f\left(  k\right)  -c\right)
V^{\prime\prime}\left(  k\right)  ,
\]
and hence, noting that $\frac{dk}{dt}=f\left(  k\right)  -c\left(  t\right)
:$%
\[
\left(  \delta-f^{\prime}\left(  k\right)  \right)  V^{\prime}\left(
k\right)  -u_{2}^{\prime}(c,k)=\left(  f\left(  k\right)  -c\right)  \frac
{d}{dk}V^{\prime}\left(  k\right)  =\frac{d}{dt}V^{\prime}\left(  k\right)  .
\]

Replacing $V^{\prime}\left(  k\right)  $ by $u_{1}^{\prime}(\sigma\left(
k\right)  ,k)=u_{1}^{\prime}(c,k)$, we get:%
\[
\left(  \delta-f^{\prime}\left(  k\right)  \right)  u_{1}^{\prime}%
(c,k)-u_{2}^{\prime}(c,k)=\frac{d}{dt}u_{1}^{\prime}(c,k),
\]
which is precisely the Euler equation.

\bigskip

\section{Time-inconsistency.\label{time-inconsistency}}

\subsection{Equilibrium strategies}

We consider the intertemporal decision problem (as it seen at time $t=0$)
\begin{equation}
J\left(  c,k\right)  =\int_{0}^{\infty}\left[  h\left(  t\right)  u\left(
c\left(  t\right)  ,k\left(  t\right)  \right)  +H\left(  t\right)  U\left(
c\left(  t\right)  ,k\left(  t\right)  \right)  \right]  dt \label{model1}%
\end{equation}
under the dynamics described by (\ref{a7}) and (\ref{a8}). Here $h$ and $H$
are discount factors, i.e. $C^{\infty}$ non-increasing functions on
$[0,\ \infty)$, such that $h\left(  0\right)  =H\left(  0\right)  =1$ and
$h\left(  \infty\right)  =H\left(  \infty\right)  =0$, while $u$ and $U$ are
utility functions. They are assumed to be $C^{\infty}$ on $]0,\ \infty)^{2}$,
with $u^{\prime\prime}<0\,$\ and $U^{\prime\prime}<0$ everywhere. We shall
also assume that they decay exponentially, so that there is some $\rho>0$ and
some $T>0$ such that $h\left(  t\right)  <e^{-\rho t}$ and $H\left(  t\right)
<e^{-\rho t}$ for $t\geq T$.

Because of time-inconsistency, the decision problem can no longer be seen as
an optimization problem. There is no way for the decision-maker at time $0$ to
achieve what is, from her point of view, the first-best solution of the
problem, and she must turn to a second-best policy: the best she can do is to
guess what her successor are planning to do, and then to plan her own
consumption $c(0)$ accordingly. In other word, we will be looking for a
subgame-perfect equilibrium of the leader-follower game played by successive generations.

The equilibrium policy was described in \cite{EkL} for the case when the
criterion (\ref{model1}) did not include the second term, and $u$ did not
depend on $k$. We will extend this analysis to the present situation, and then
compare it with the approach in \cite{Kar2}.

\begin{definition}
A Markov strategy $c=\sigma\left(  k\right)  $ is \emph{convergent} if there
is a point $k_{\infty}$ and a neighbourhood $\mathcal{N}$ of $k_{\infty}$ such
that, for every $k_{0}\in\mathcal{N}$ the solution $k(t)$ of (\ref{20c})
converges to $k_{\infty}$ (and so $c\left(  t\right)  =\sigma\left(  k\left(
t\right)  \right)  $ converges to $c_{\infty}=f\left(  k_{\infty}\right)  $).
If this is the case, we shall say that $\left(  c_{\infty},k_{\infty}\right)
$ is \emph{supported} by $\sigma$.
\end{definition}

Let us suppose that a convergent Markov strategy $\sigma$ has been announced
and is the public knowledge $\sigma$. The decision-maker at $T$ has capital
stock $k_{T}$. If all future decision-maker apply the strategy $\sigma$, the
resulting future capital stock flow $k(t)$ obeys:
\begin{equation}
\frac{dk}{dt}=f\left(  k\right)  -\sigma\left(  k\right)  ,\ k\left(
T\right)  =k_{T},\ T\leq t \label{equation.of.kb}%
\end{equation}

Since every decision-maker faces the same problem (with different stock
levels) it is enough to take $T=0$. Suppose the decision-maker at time $0$
holds power for $0\leq t<\varepsilon$, and expects all later decision-makers
to apply the strategy $\sigma$. He or she then explores whether it is in his
or her interest to apply the strategy $\sigma$, that is, to play $c_{0}%
=\sigma\left(  k_{0}\right)  $ for $0\leq t<\varepsilon$. If he or she applies
the constant control $c$ for $0\leq t\leq\varepsilon$.

Suppose the constant control $c$ is use on $0\leq t\leq\varepsilon$. The
immediate utility flow during $[0,\varepsilon]$ is $[u(c,k_{0})+U(c,k_{0}%
)]\varepsilon+o(\varepsilon)$ where $o(\varepsilon)$ is a higher order term of
$\varepsilon$. At time $\varepsilon$, the resulting capital will be
$k_{0}+(f(k_{0})-c)\varepsilon+o(\varepsilon)$. From then on, the strategy
$\sigma$ will be applied, which results in a capital stock $k_{c}$ satisfying
(we omit higher-order terms):
\begin{equation}
\frac{dk_{c}}{dt}=f\left(  k_{c}\right)  -\sigma\left(  k_{c}\right)
,\ k_{c}\left(  \varepsilon\right)  =k_{0}+(f(k_{0})-c)\varepsilon
,\ t\geq\varepsilon\label{kc}%
\end{equation}

The capital stock $k_{c}$ can be written as
\begin{equation}
k_{c}(t)=k_{0}(t)+k_{1}(t)\varepsilon, \label{kc.kb.ki}%
\end{equation}
where $k_{0}\left(  t\right)  $ is the unperturbed solution, and $k_{1}\left(
t\right)  $ is given by the linearized equation:%
\begin{align}
\frac{dk_{0}}{dt}  &  =f\left(  k_{0}\right)  -\sigma\left(  k_{0}\right)
,\text{ }k_{0}\left(  0\right)  =k_{0}\label{a12}\\
\frac{dk_{1}}{dt}  &  =\left(  f^{\prime}\left(  k_{0}\right)  -\sigma
^{\prime}\left(  k_{0}\right)  \right)  k_{1},\ \ k_{1}\left(  0\right)
=\sigma\left(  k_{0}\right)  -c\ \label{a11}%
\end{align}

Evaluating the integral (\ref{model1}) we get:%

\[%
\begin{array}
[l]{c}%
J\left(  \varepsilon\right)  =u(c,k_{0})\varepsilon+\int_{\varepsilon}%
^{\infty}h(s)u(\sigma(k_{0}(t)+\varepsilon k_{1}(t)),k_{0}(t)+\varepsilon
k_{1}(t))dt\\
\quad\quad\quad\quad\ +U(c,k_{0})\varepsilon+\int_{\varepsilon}^{\infty
}H(s)U(\sigma(k_{0}(t)+\varepsilon k_{1}(t)),k_{0}(t)+\varepsilon k_{1}(t))dt
\end{array}
\]

Letting $\varepsilon\rightarrow0$, so that commitment span of the
decision-maker vanishes, we get:%
\[
\lim_{\varepsilon\rightarrow0}\frac{1}{\varepsilon}\left(  J\left(
\varepsilon\right)  -J\left(  0\right)  \right)  ={P}(k_{0},\sigma,c)
\]
where
\begin{align}
P(k_{0},\sigma,c)  &  =u(c,k_{0})-u(\sigma(k_{0}),k_{0})+(U(c,k_{0}%
)-U(\sigma(k_{0}),k_{0}))\nonumber\\
&  +\int_{0}^{\infty}h(t)u_{1}^{\prime}(\sigma(k_{0}(t)),k_{0}(t))\sigma
^{\prime}(k_{0}(t))k_{1}(t)dt+\int_{0}^{\infty}h(t)u_{2}^{\prime}(\sigma
(k_{0}(t)),k_{0}(t))k_{1}(t)dt\nonumber\\
&  +\int_{0}^{\infty}H(t)U_{1}^{\prime}(\sigma(k_{0}(t)),k_{0}(t))\sigma
^{\prime}(k_{0}(t))k_{1}(t)dt+\int_{0}^{\infty}H(t)U_{2}^{\prime}(\sigma
(k_{0}(t)),k_{0}(t))k_{1}(t)dt,\nonumber
\end{align}

\begin{definition}
A convergent Markov strategy $\sigma$ is an equilibrium if we have:
\[
\max_{c}P(k,\sigma,c)=P(k,\sigma,\sigma\left(  k\right)  ),\text{ \ }\forall
k.
\]

\end{definition}

\subsection{The HJB\ approach}

We now characterizes the equilibrium strategy. We write $k_{c}(t)=\mathcal{K}%
(t;k_{0},\sigma)$ where $\mathcal{K}$ is the flow associated with the
differential equation (\ref{a12}). We also define a function $\varphi\,$\ by:%
\begin{align*}
u_{1}^{\prime}(\varphi(x,k),k)+U_{1}^{\prime}\left(  \varphi(x,k),k\right)
&  =x,\\
\varphi(u_{1}^{\prime}(c,k)+U_{1}^{\prime}\left(  c,k),k\right)   &  =c
\end{align*}

Since $u_{1}^{\prime\prime}$ and $U_{1}^{\prime\prime}$ are both negative, the
function $\varphi$ is well-defined by the implicit function theorem.

\begin{theorem}
\label{theorem.ie} Let $\sigma$ be an equilibrium strategy. The function:
\begin{equation}
V(k_{0})=\int_{0}^{\infty}h(t)u(\sigma(\mathcal{K}(t;k_{0},\sigma
)),\mathcal{K}(t;k_{0},\sigma))dt+\int_{0}^{\infty}H(t)U(\sigma(\mathcal{K}%
(t;k_{0},\sigma)),\mathcal{K}(t;k_{0},\sigma))dt \label{value.function}%
\end{equation}
satisfies the integral equation:%
\begin{equation}%
\begin{array}
[c]{l}%
V(k_{0})=\int_{0}^{\infty}h(t)u(\varphi\circ V^{\prime}(\mathcal{K}%
(t;k_{0},\varphi\circ V^{\prime})),\mathcal{K}(t;k_{0},\varphi\circ V^{\prime
}))dt\\
\quad\quad\quad\ \ +\int_{0}^{\infty}H(t)U(\varphi\circ V^{\prime}%
(\mathcal{K}(t;k_{0},\varphi\circ V^{\prime})),\mathcal{K}(t;k_{0}%
,\varphi\circ V^{\prime}))dt
\end{array}
\tag{IE}\label{IE}%
\end{equation}
and the instantaneous optimality condition
\begin{equation}
u_{1}^{\prime}(\sigma(k_{0}),k_{0})+U_{1}^{\prime}(\sigma(k_{0}),k_{0}%
)=V^{\prime}(k_{0}),\quad\sigma(k_{0})=\varphi(V^{\prime}(k_{0}),k_{0}),
\label{instantaneous.condition}%
\end{equation}

Conversely, suppose a function $V$ is twice continuously differentiable,
satisfies (IE), and the strategy $\sigma(k_{0}):=\varphi(V^{\prime}%
(k_{0}),k_{0})$ is convergent. Then $\sigma$ is an equilibrium strategy.
\end{theorem}

For the sake of convenience, we have shortened $\sigma(k_{0}):=\varphi
(V^{\prime}(k_{0}),k_{0})$ to $\sigma=\varphi\circ V^{\prime}$.

\begin{proof}
Since the system is autonomous, we have:
\begin{equation}
\mathcal{K}(s;\mathcal{K}(t;k_{0},\sigma),\sigma)=\mathcal{K}(s+t;k_{0}%
,\sigma). \label{property.1.of.K}%
\end{equation}

Next, denote the fundamental solution of the linearized equation of
(\ref{equation.of.kb}) at $k_{0}$ by $\mathcal{R}(k_{0};t)$ so that:
\begin{align*}
k_{1}(t)  &  =\mathcal{R}(k_{0};t)(\sigma(k_{0})-c)\\
\frac{d\mathcal{R}}{dt}  &  =(f^{\prime}(\mathcal{K}(t;k_{0},\sigma
))-\sigma^{\prime}(\mathcal{K}(t;k_{0},\sigma)))\mathcal{R}(t),\quad
\mathcal{R}(k_{0};0)=I,
\end{align*}

$\mathcal{R}$ and $\mathcal{K}$ are related by:
\begin{align}
\frac{\partial\mathcal{K}(t;k_{0},\sigma)}{\partial k_{0}}  &  =\mathcal{R}%
(k_{0};t),\label{property.2.of.K}\\
\frac{\partial\mathcal{K}(t;k_{0},\sigma)}{\partial t}  &  =f(\mathcal{K}%
(t;k_{0},\sigma))-\sigma(\mathcal{K}(t;k_{0},\sigma)). \label{property.3.of.K}%
\end{align}

Let us now turn to the first part of the theorem. Differentiating
(\ref{value.function})with respect to $k_{0}$:%
\begin{align}
V^{\prime}(k_{0})  &  =\int_{0}^{\infty}h(t)u_{1}^{\prime}(\sigma
(\mathcal{K}(t;k_{0},\sigma)),\mathcal{K}(t;k_{0},\sigma))\sigma^{\prime
}(\mathcal{K}(t;k_{0},\sigma))\mathcal{R}(k_{0};t)dt\nonumber\\
&  +\int_{0}^{\infty}h(t)u_{2}^{\prime}(\sigma(\mathcal{K}(t;k_{0}%
,\sigma)),\mathcal{K}(t;k_{0},\sigma))\mathcal{R}(k_{0};t)dt\nonumber\\
&  +\int_{0}^{\infty}H(t)U_{1}^{\prime}(\sigma(\mathcal{K}(t;k_{0}%
,\sigma)),\mathcal{K}(t;k_{0},\sigma))\sigma^{\prime}(\mathcal{K}%
(t;k_{0},\sigma))\mathcal{R}(k_{0};t)dt\nonumber\\
&  +\int_{0}^{\infty}H(t)U_{2}^{\prime}(\sigma(\mathcal{K}(t;k_{0}%
,\sigma)),\mathcal{K}(t;k_{0},\sigma))\mathcal{R}(k_{0};t)dt. \label{c5}%
\end{align}

Substituting $k_{0}(t)=\mathcal{K}(t;k_{0},\sigma)$ and $k_{1}(t)$ in the
definition of $P$, we get:%
\begin{align}
P(k_{0},\sigma,c)  &  =\lbrack u(c,k_{0})+U(c,k_{0})]-[u(\sigma(k_{0}%
),k_{0})+U(\sigma(k_{0}),k_{0})]\nonumber\\
&  +\int_{0}^{\infty}h(t)u_{1}^{\prime}(\sigma(\mathcal{K}(t;k_{0}%
,\sigma)),\mathcal{K}(t;k_{0},\sigma))\sigma^{\prime}(\mathcal{K}%
(t;k_{0},\sigma))\mathcal{R}(k_{0};t)(\sigma(k_{0})-c)dt\nonumber\\
&  +\int_{0}^{\infty}h(t)u_{2}^{\prime}(\sigma(\mathcal{K}(t;k_{0}%
,\sigma)),\mathcal{K}(t;k_{0},\sigma))\mathcal{R}(k_{0};t)(\sigma
(k_{0})-c)dt\nonumber\\
&  +\int_{0}^{\infty}H(t)U_{1}^{\prime}(\sigma(\mathcal{K}(t;k_{0}%
,\sigma)),\mathcal{K}(t;k_{0},\sigma))\sigma^{\prime}(\mathcal{K}%
(t;k_{0},\sigma))\mathcal{R}(k_{0};t)(\sigma(k_{0})-c)dt\nonumber\\
&  +\int_{0}^{\infty}H(t)U_{2}^{\prime}(\sigma(\mathcal{K}(t;k_{0}%
,\sigma)),\mathcal{K}(t;k_{0},\sigma))\mathcal{R}(k_{0};t)(\sigma
(k_{0})-c)dt.\nonumber
\end{align}

Since $u$ and $U$ are strictly concave and differentiable with respect to $c$,
the necessary and sufficient condition to maximize $P(k_{0},\sigma,c)$ with
respect to $c$ is that the derivative vanishes at $c=\sigma(k_{0})$, that is:
\begin{align*}
u_{1}^{\prime}(\sigma(k_{0}),k_{0})+U_{1}^{\prime}(\sigma(k_{0}),k_{0})=  &
\int_{0}^{\infty}h(t)u_{1}^{\prime}(\sigma(\mathcal{K}(t;k_{0},\sigma
)),\mathcal{K}(t;k_{0},\sigma))\sigma^{\prime}(\mathcal{K}(t;k_{0}%
,\sigma))\mathcal{R}(k_{0};t)dt\\
&  +\int_{0}^{\infty}h(t)u_{2}^{\prime}(\sigma(\mathcal{K}(t;k_{0}%
,\sigma)),\mathcal{K}(t;k_{0},\sigma))\mathcal{R}(k_{0};t)dt\\
&  +\int_{0}^{\infty}H(t)U_{1}^{\prime}(\sigma(\mathcal{K}(t;k_{0}%
,\sigma)),\mathcal{K}(t;k_{0},\sigma))\sigma^{\prime}(\mathcal{K}%
(t;k_{0},\sigma))\mathcal{R}(k_{0};t)dt\\
&  +\int_{0}^{\infty}H(t)U_{2}^{\prime}(\sigma(\mathcal{K}(t;k_{0}%
,\sigma)),\mathcal{K}(t;k_{0},\sigma))\mathcal{R}(k_{0};t)dt,
\end{align*}

The right-hand side is precisely $V^{\prime}(k_{0})$, as we wanted. Therefore,
the equilibrium strategy satisfies
\[
u_{1}^{\prime}(\sigma(k_{0}),k_{0})+U_{1}^{\prime}(\sigma(k_{0}),k_{0}%
)=V^{\prime}(k_{0})
\]
and we have $\sigma(k_{0})=\varphi(V^{\prime}(k_{0}),k_{0})$. Substituting
back into equation (\ref{value.function}), we get the functional equation
(IE). This prove the first part of the theorem (necessity). We refer to
\cite{EkL} for the second part (sufficiency).
\end{proof}

The following theorem gives an alternative characterization, the differential
equation, which resembles the usual HJB equation from the calculus of variation.

\begin{theorem}
\label{theorem.de} Let $V$ be a $C^{2}$ function such that the strategy
$\sigma=\varphi\circ V^{\prime}$ converges to $\bar{k}$. Then $V$ satisfies
the integral equation (\ref{IE}) if and only if it satisfies the following
integro-differential equation%
\begin{equation}%
\begin{array}
[c]{c}%
u(\varphi\circ{V}^{\prime}(k_{0}),k_{0})+U(\varphi\circ{V}^{\prime}%
(k_{0}),k_{0})+V^{\prime}(k_{0})(f(k_{0})-\varphi({V}^{\prime}(k_{0}%
),k_{0}))\\
=-\int_{0}^{\infty}h^{\prime}(t)u(\varphi\circ{V}^{\prime}(\mathcal{K}%
(t;k_{0},\varphi\circ V^{\prime})),\mathcal{K}(t;k_{0},\varphi\circ V^{\prime
}))dt\\
-\int_{0}^{\infty}H^{\prime}(t)U(\varphi\circ{V}^{\prime}(\mathcal{K}%
(t;k_{0},\varphi\circ V^{\prime})),\mathcal{K}(t;k_{0},\varphi\circ V^{\prime
}))dt
\end{array}
\tag{DE}\label{DE}%
\end{equation}
together with the boundary condition
\begin{equation}
V(\bar{k})=u(f(\bar{k}),\bar{k})\int_{0}^{\infty}h(t)dt+U(f(\bar{k}),\bar
{k})\int_{0}^{\infty}H(t)dt. \tag{BC}\label{BC}%
\end{equation}

\end{theorem}

\begin{proof}
Introduce the function $\phi$ defined by
\begin{equation}
\phi(k_{0})=V(k_{0})-\int_{0}^{\infty}h(t)u(\sigma(\mathcal{K}(t;k_{0}%
,\sigma)),\mathcal{K}(t;k_{0},\sigma))dt-\int_{0}^{\infty}H(t)U(\sigma
(\mathcal{K}(t;k_{0},\sigma)),\mathcal{K}(t;k_{0},\sigma))dt,\nonumber
\end{equation}
where $\sigma(k_{0})=\varphi(V^{\prime}(k_{0}),k_{0})$. Consider the value
$\psi(t,k_{0})$ of $\phi$ along the trajectory $t\rightarrow\mathcal{K}%
(t;k_{0},\sigma)$ originating from $k_{0}$ at time $0$, that is
\begin{align*}
\psi(t,k_{0})  &  =\phi(\mathcal{K}(t;k_{0},\sigma))\\
&  =V(\mathcal{K}(t;k_{0},\sigma))-\int_{t}^{\infty}h(s-t)u(\sigma
(\mathcal{K}(\sigma;s,k_{0})),\mathcal{K}(\sigma;s,k_{0}))ds\\
&  \quad-\int_{t}^{\infty}H(s-t)U(\sigma(\mathcal{K}(\sigma;s,k_{0}%
)),\mathcal{K}(\sigma;s,k_{0}))ds,
\end{align*}

Then the derivative of $\psi$ with respect to $t$ is given by:
\begin{align}
\frac{\partial\psi(t,k_{0})}{\partial{t}}=  &  V^{\prime}(\mathcal{K}%
_{t})[f(\mathcal{K}_{t})-\varphi\circ{V^{\prime}}(\mathcal{K}_{t})]
+u(\varphi\circ{V^{\prime}}(\mathcal{K}_{t}),\mathcal{K}_{t})+U(\varphi
\circ{V^{\prime}}(\mathcal{K}_{t}),\mathcal{K}_{t})\nonumber\\
&  +\int_{0}^{\infty}h^{\prime}(s)u(\varphi\circ{V^{\prime}}(\mathcal{K}%
(s;\mathcal{K}_{t},\varphi\circ{V^{\prime}})),\mathcal{K}(s;\mathcal{K}%
_{t},\varphi\circ{V^{\prime}}))ds\nonumber\\
&  +\int_{0}^{\infty}H^{\prime}(s)U(\varphi\circ{V^{\prime}}(\mathcal{K}%
(s;\mathcal{K}_{t},\varphi\circ{V^{\prime}})),\mathcal{K}(s;\mathcal{K}%
_{t},\varphi\circ{V^{\prime}}))ds,\nonumber
\end{align}
where we have denoted $\mathcal{K}(t;k_{0},\varphi\circ V^{\prime})$ by
$\mathcal{K}_{t}$ for convenience. If (\ref{DE}) holds, then the right hand
side is identically zero along the trajectory, so that $\psi(t,k_{0})$ does
not depend on $t$, thus $\psi(s,k_{0})=\psi(t,k_{0})$ for all $s,t\geq0$.
Letting $t\rightarrow{\infty}$ in the definition of $\psi$, we get
\begin{align}
\psi(s,k_{0})  &  =\lim_{t\rightarrow{\infty}}\psi(t,k_{0})\nonumber\\
&  =V(\bar{k})-\int_{0}^{\infty}h(s)u(\sigma(\bar{k}),\bar{k})ds-\int%
_{0}^{\infty}H(s)U(\sigma(\bar{k}),\bar{k})ds,
\end{align}
and hence, if (BC) holds, then $\psi=\phi\equiv0$ and so equation (IE) holds.
Conversely, if $V(k)$ satisfies equation (IE), then the same lines of
reasoning shows that equation (DE) and the boundary condition (BC) are satisfied.
\end{proof}

\subsection{The Euler equations}

To obtain the Euler-Lagrange-like equation, we differentiate the both side of
(\ref{DE}) with respect to $k_{0}$. We get:
\begin{align*}
&  -\int_{0}^{\infty}h^{\prime}(t)u_{1}^{\prime}(\sigma(\mathcal{K}%
(t;k_{0},\sigma)),\mathcal{K}(t;k_{0},\sigma))\sigma^{\prime}(\mathcal{K}%
(t;k_{0},\sigma))\frac{\partial\mathcal{K}(t;k_{0},\sigma)}{\partial k_{0}%
}dt\\
&  -\int_{0}^{\infty}h^{\prime}(t)u_{2}^{\prime}(\sigma(\mathcal{K}%
(t;k_{0},\sigma)),\mathcal{K}(t;k_{0},\sigma))\frac{\partial\mathcal{K}%
(t;k_{0},\sigma)}{\partial k_{0}}dt\\
&  -\int_{0}^{\infty}H^{\prime}(t)U_{1}^{\prime}(\sigma(\mathcal{K}%
(t;k_{0},\sigma)),\mathcal{K}(t;k_{0},\sigma))\sigma^{\prime}(\mathcal{K}%
(t;k_{0},\sigma))\frac{\partial\mathcal{K}(t;k_{0},\sigma)}{\partial k_{0}%
}dt\\
&  -\int_{0}^{\infty}H^{\prime}(t)U_{2}^{\prime}(\sigma(\mathcal{K}%
(t;k_{0},\sigma)),\mathcal{K}(t;k_{0},\sigma))\frac{\partial\mathcal{K}%
(t;k_{0},\sigma)}{\partial k_{0}}dt\\
&  =[u_{1}^{\prime}(\sigma(k_{0}),k_{0})+U_{1}^{\prime}(\sigma(k_{0}%
),k_{0})]\sigma^{\prime}(k_{0})+u_{2}^{\prime}(\sigma(k_{0}),k_{0}%
)+U_{2}^{\prime}(\sigma(k_{0}),k_{0})\\
&  \quad+V^{\prime}(k_{0})f^{\prime}(k_{0})-V^{\prime}(k_{0})\sigma^{\prime
}(k_{0})+V^{\prime\prime}(k_{0})(f(k_{0})-\sigma(k_{0})).
\end{align*}

Plugging $k_{0}=k(t)$, $\sigma(k(t))=c(t)$ and using
(\ref{instantaneous.condition}) to cancel the first and the fifth terms,
together with (\ref{24}) and (\ref{property.2.of.K}), we have
\begin{align*}
&  -\int_{0}^{\infty}h^{\prime}(s)u_{1}^{\prime}(\sigma(\mathcal{K}%
(s;k(t),\sigma)),\mathcal{K}(s;k(t),\sigma))\sigma^{\prime}(\mathcal{K}%
(s;k(t),\sigma))\mathcal{R}(k\left(  t\right)  ,s)ds\\
&  -\int_{0}^{\infty}h^{\prime}(s)u_{2}^{\prime}(\sigma(\mathcal{K}%
(s;k(t),\sigma)),\mathcal{K}(s;k(t),\sigma))\mathcal{R}(k\left(  t\right)
,s)ds\\
&  -\int_{0}^{\infty}H^{\prime}(s)U_{1}^{\prime}(\sigma(\mathcal{K}%
(s;k(t),\sigma)),\mathcal{K}(s;k(t),\sigma))\sigma^{\prime}(\mathcal{K}%
(s;k(t),\sigma))\mathcal{R}(k\left(  t\right)  ,s)ds\\
&  -\int_{0}^{\infty}H^{\prime}(s)U_{2}^{\prime}(\sigma(\mathcal{K}%
(s;k(t),\sigma)),\mathcal{K}(s;k(t),\sigma))\mathcal{R}(k\left(  t\right)
,s)ds\\
&  =u_{2}^{\prime}(c(t),k(t))+U_{2}^{\prime}(c(t),k(t))+V^{\prime
}(k(t))f^{\prime}(k(t))+V^{\prime\prime}(k(t))(f(k(t))-c(t))\\
&  =u_{2}^{\prime}(c(t),k(t))+U_{2}^{\prime}(c(t),k(t))+[u_{1}^{\prime
}(c(t),k(t))+U_{1}^{\prime}(c(t),k(t))]f^{\prime}(k(t))\\
&  \quad+\frac{d}{dt}[u_{1}^{\prime}(c(t),k(t))+U_{1}^{\prime}(c(t),k(t))],
\end{align*}
where to get the last two terms in the right hand side, we have used
(\ref{instantaneous.condition}) again.

We have $\mathcal{K}(s;k(t),\sigma)=k(s+t)$, and $\mathcal{R}(k\left(
t\right)  ;s)=\mathcal{R}(k_{0};s+t)$
\begin{align}
\label{sigma'.of.kc}\sigma^{\prime}\left(  k\left(  t\right)  \right)   &
=\frac{1}{f(k(t))-c(t)}\frac{dc}{dt}\\
\mathcal{R}(k_{0};s+t)  &  =\exp(f^{\prime}(k\left(  t\right)  )-\frac
{1}{f(k(t))-c(t)}\frac{dc}{dt}),\nonumber
\end{align}

Writing this into the preceding equations, we finally get:
\begin{align}
&  -\int_{t}^{\infty}h^{\prime}(s-t)\left[  u_{1}^{\prime}(c(s),k(s))\gamma
\left(  s\right)  +u_{2}^{\prime}(c(s),k(s))\right]  e^{f^{\prime
}(k(s))-\gamma\left(  s\right)  }ds\label{a15}\\
&  -\int_{t}^{\infty}H^{\prime}(s-t)\left[  U_{1}^{\prime}(c(s),k(s))\gamma
\left(  s\right)  +U_{2}^{\prime}(c(s),k(s))\right]  e^{f^{\prime
}(k(s))-\gamma\left(  s\right)  }ds\label{a16}\\
&  =u_{2}^{\prime}(c(t),k(t))+U_{2}^{\prime}(c(t),k(t))+\left[  u_{1}^{\prime
}(c(t),k(t))+U_{1}^{\prime}(c(t),k(t))\right]  f^{\prime}(k(t))\label{a17}\\
&  +\frac{d}{dt}[u_{1}^{\prime}(c(t),k(t))+U_{1}^{\prime}(c(t),k(t))]
\label{a18}%
\end{align}
with:%
\begin{equation}
\gamma\left(  s\right)  :=\frac{1}{f(k(s))-c(s)}\frac{dc}{ds} \label{a19}%
\end{equation}
which is the Euler-Lagrange-like equation for the time-inconsistent case.

\subsection{The control theory approach}

Karp \cite{Kar2}, \cite{KaL1} has developed a different method to deal with
time-inconsistency. In this section, we connect his results with ours.

Defining $V(k_{0})$ as above, we must have:
\begin{equation}
V(k_{0})=\max_{c}\left\{  u(c,k_{0})+U(c,k_{0})]\varepsilon+\int_{\varepsilon
}^{\infty}[h(t)u(\sigma(k_{c}(t)),k_{c}(t))+H(t)U(\sigma(k_{c}(t)),k_{c}%
(t))]dt\right\}  \label{karp1}%
\end{equation}

On the other hand, we have
\[
V(k_{c}(\varepsilon))=\int_{\varepsilon}^{\infty}[h(s-\varepsilon
)u(\sigma(k_{c}(s)),k_{c}(s))+H(s-\varepsilon)U(\sigma(k_{c}(s)),k_{c}%
(s))]ds,
\]

Substituting:
\begin{equation}
V(k_{0})=\max_{c}\left\{
\begin{array}
[c]{l}%
\lbrack u(c,k_{0})+U(c,k_{0})]\varepsilon+V(k_{c}(\varepsilon))\\
+\int_{\varepsilon}^{\infty}[h(t)-h(t-\varepsilon)]u(\sigma(k_{c}%
(t)),k_{c}(t))dt\\
+\int_{\varepsilon}^{\infty}[H(t)-H(t-\varepsilon)]U(\sigma(k_{c}%
(t)),k_{c}(t))dt
\end{array}
\right\}  \nonumber\label{karp3}%
\end{equation}
Letting $\varepsilon\rightarrow0$, we have
\begin{align}
&  -\int_{0}^{\infty}h^{\prime}(t)u(\sigma(k(t)),k(t))dt-\int_{0}^{\infty
}H^{\prime}(t)U(\sigma(k(t)),k(t))dt\nonumber\\
&  =\max_{c}\{u(c,k_{0})+U(c,k_{0})+V^{\prime}(k_{0})(f(k_{0})-c)\}. \label{K}%
\end{align}

This equation was first obtained by Karp \cite{Kar2}, \cite{KaL1}. This is the
HJB equation for a certain control problem, which he terms the auxiliary
control problem. We now show that this approach is equivalent to the preceding
one. approach by Ekeland-Lazrak. In the right hand side of (\ref{K}), the
maximum is attained at:
\[
c=\ \arg\max_{c}\{u(c,k_{0})+U(c,k_{0})+V^{\prime}(k_{0})(f(k_{0}%
)-c)\}=\varphi(V^{\prime}(k_{0}),k_{0}),
\]
thus we have
\begin{align*}
&  -\int_{0}^{\infty}h^{\prime}(t)u(\sigma(k(t)),k(t))dt-\int_{0}^{\infty
}H^{\prime}(t)U(\sigma(k(t)),k(t))dt\\
&  =u(\varphi(V^{\prime}(k_{0}),k_{0}),k_{0})+U(\varphi(V^{\prime}%
(k_{0}),k_{0}),k_{0})+V^{\prime}(k_{0})(f(k_{0})-\varphi(V^{\prime}%
(k_{0}),k_{0})).
\end{align*}
which is exactly the same as (DE)

\section{The biexponential case}

\subsection{The equations}

In this section, we consider the biexponential criterion
\begin{equation}
\lambda\int_{0}^{\infty}e^{-\delta_{1}s}u(c(s),k(s))ds+(1-\lambda)\int%
_{0}^{\infty}e^{-\delta_{2}s}U(c(s),k(s))ds. \label{biexponential.criterion}%
\end{equation}

Without loss of generality, we assume that:%
\[
\delta_{1}>\delta_{2}%
\]

If $\lambda=0$ or $1$, then this is just the Ramsey criterion (\ref{18}). Thus
we are interested in the case $0<\lambda<1$. Dividing by $\lambda$, we find
that (\ref{biexponential.criterion}) is a special case of (\ref{model1}),
where $h(t)=e^{-\delta_{1}t}$, $H(t)=e^{-\delta_{2}t}$ and $U$ is replaced by
$\frac{1-\lambda}{\lambda}U$. So all the results of the preceding section hold.

\subsubsection{The HJB-type equations}

Given an equilibrium strategy $\sigma_{\lambda}$, introduce the two functions:%
\begin{align}
V_{\lambda}\left(  k\right)   &  :=\int_{0}^{\infty}e^{-\delta_{1}t}%
u(\sigma_{\lambda}\left(  k(t)\right)  ,k(t))dt+\frac{1-\lambda}{\lambda}%
\int_{0}^{\infty}e^{-\delta_{2}t}U(\sigma_{\lambda}\left(  k(t)\right)
,k(t))dt\label{25a}\\
W_{\lambda}\left(  k\right)   &  :=\int_{0}^{\infty}e^{-\delta_{1}t}%
u(\sigma_{\lambda}\left(  k(t)\right)  ,k(t))dt-\frac{1-\lambda}{\lambda}%
\int_{0}^{\infty}e^{-\delta_{2}t}U(\sigma_{\lambda}\left(  k(t)\right)
,k(t))dt \label{26a}%
\end{align}

In Proposition \ref{equivalent.to.an.ODE.system}, we prove that the HJB-type
equation (DE)\ reduces to a system of two ODEs for $V_{\lambda}$ and
$W_{\lambda}:$%
\begin{align}
\left(  f-\varphi_{\lambda}\left(  V_{\lambda}^{\prime}\right)  \right)
V_{\lambda}^{\prime}+u\left(  V_{\lambda}^{\prime}\right)  +\frac{1-\lambda
}{\lambda}U\left(  V_{\lambda}^{\prime}\right)  =  &  \frac{\delta_{1}%
+\delta_{2}}{2}V_{\lambda}+\frac{\delta_{1}-\delta_{2}}{2}W_{\lambda
},\label{35}\\
\left(  f-\varphi_{\lambda}\left(  V_{\lambda}^{\prime}\right)  \right)
W_{\lambda}^{\prime}+u\left(  V_{\lambda}^{\prime}\right)  -\frac{1-\lambda
}{\lambda}U\left(  V_{\lambda}^{\prime}\right)  =  &  \frac{\delta_{1}%
-\delta_{2}}{2}V_{\lambda}+\frac{\delta_{1}+\delta_{2}}{2}W_{\lambda}
\label{36}%
\end{align}
where $\varphi_{\lambda}\left(  V_{\lambda}^{\prime}\right)  $ and $u\left(
V_{\lambda}^{\prime}\right)  $ denote the functions $k\rightarrow
\varphi_{\lambda}\left(  V_{\lambda}^{\prime}\left(  k\right)  ,k\right)  $
and $k\rightarrow u\left(  \varphi_{\lambda}\left(  V^{\prime}\right)
,k\right)  $. Recall that $\varphi_{\lambda}$ is defined by the equivalent
equations:%
\begin{align}
u_{1}^{\prime}(\varphi_{\lambda}(x,k),k)+\frac{1-\lambda}{\lambda}%
U_{1}^{\prime}\left(  \varphi_{\lambda}(x,k),k\right)   &  =x,\label{39}\\
\varphi_{\lambda}(u_{1}^{\prime}(c,k)+\frac{1-\lambda}{\lambda}U_{1}^{\prime
}\left(  c,k),k\right)   &  =c \label{40}%
\end{align}

Similarly, the boundary condition (BC) becomes:%
\begin{align}
V_{\lambda}\left(  k_{\infty}\right)   &  =\frac{1}{\delta_{1}}u\left(
f\left(  k_{\infty}\right)  ,k_{\infty}\right)  +\frac{1-\lambda}{\lambda
}\frac{1}{\delta_{2}}U\left(  f\left(  k_{\infty}\right)  ,k_{\infty}\right)
,\label{37}\\
W_{\lambda}\left(  k_{\infty}\right)   &  =\frac{1}{\delta_{1}}u\left(
f\left(  k_{\infty}\right)  ,k_{\infty}\right)  -\frac{1-\lambda}{\lambda
}\frac{1}{\delta_{2}}U\left(  f\left(  k_{\infty}\right)  ,k_{\infty}\right)
. \label{38}%
\end{align}

If $\lambda=1$, we find the usual HJB equation (\ref{20}) for $V$ with the
boundary condition (\ref{20a}).

\begin{proposition}
\label{equivalent.to.an.ODE.system} \ Suppose that $\sigma_{\lambda}$ is an
equilibrium strategy such that $k\left(  t\right)  $ converges to $k_{\infty}%
$. Then $V_{\lambda}$ and $W_{\lambda}$ defined by (\ref{25a}) and (\ref{26a})
satisfy the equations (\ref{35}) and (\ref{36}) with the boundary conditions
(\ref{37}) and (\ref{38}). Conversely, suppose there is a point $k_{\infty}$,
a $C^{2}$ function $V_{\lambda}$ and a $C^{1}$ function $W_{\lambda}$, both
defined on some open neighbourhood of $k_{\infty}$, satisfying the equations
(\ref{35}) and (\ref{36}) with the boundary conditions (\ref{37}) and
(\ref{38}); suppose moreover that strategy $\sigma_{\lambda}\left(  k\right)
:=\varphi_{\lambda}\left(  V_{\lambda}^{\prime}\left(  k\right)  ,k\right)  $
converges to $\left(  f\left(  k_{\infty}\right)  ,k_{\infty}\right)  $. Then
$\sigma_{\lambda}$ is an equilibrium strategy.
\end{proposition}

Let us draw the reader's attention to the fact that $V_{\lambda}$ must be
$C^{2}$ while $W_{\lambda}$ needs only be $C^{1}$.

\begin{proof}
Let us simplify the notation. Write $\sigma$ instead of $\sigma_{\lambda}$and
set%
\[
a=\frac{\delta_{1}+\delta_{2}}{2},\ \ b=\frac{\delta_{1}-\delta_{2}}{2}%
\]

Arguing as in Theorem \ref{theorem.de}, we obtain (\ref{35}) and
(\ref{36})\thinspace\ by differentiating (\ref{25a}) and (\ref{26a}). The
boundary conditions (\ref{37}) and (\ref{38}) follow from setting $k\left(
t\right)  =k_{\infty}$ and $c\left(  t\right)  =\sigma\left(  k_{\infty
}\right)  $ in (\ref{35}) and (\ref{36}).

Conversely suppose $v_{1}$ and $w_{1}$ satisfy (\ref{35}), (\ref{36}) and
(\ref{37}), (\ref{38}), and suppose the strategy $\sigma_{1}=\varphi_{\lambda
}\circ{v_{1}^{\prime}}$ converges to $k$. Consider the following functions
\begin{align}
v_{2}(k_{0}) =  &  \int_{0}^{\infty}e^{-\delta_{1}{t}}u(\sigma_{1}%
(\mathcal{K}(t;k_{0},\sigma_{1})),\mathcal{K}(t;k_{0},\sigma_{1}%
))dt\nonumber\\
&  +\frac{1-\lambda}{\lambda}\int_{0}^{\infty}e^{-\delta_{2}t}U(\sigma
_{1}(\mathcal{K}(t;k_{0},\sigma_{1})),\mathcal{K}(t;k_{0},\sigma
_{1}))dt,\nonumber\\
w_{2}(k_{0}) =  &  \int_{0}^{\infty}e^{-\delta_{1}{t}}u(\sigma_{1}%
(\mathcal{K}(t;k_{0},\sigma_{1})),\mathcal{K}(t;k_{0},\sigma_{1}%
))dt\nonumber\\
&  -\frac{1-\lambda}{\lambda}\displaystyle\int_{0}^{\infty}e^{-\delta_{2}%
t}U(\sigma_{1}(\mathcal{K}(t;k_{0},\sigma_{1})),\mathcal{K}(t;k_{0},\sigma
_{1}))dt.\nonumber
\end{align}

Arguing as in Theorem \ref{theorem.de}, we find that $v_{2}$ and $w_{2}$ also
satisfy (\ref{35}), (\ref{36}) and (\ref{37}), (\ref{38}). Setting
$v_{3}=v_{1}-v_{2},w_{3}=w_{1}-w_{2}$,
then we get
\begin{equation}%
\begin{array}
[c]{c}%
\left(  f\left(  k\right)  -\sigma\left(  k\right)  \right)  v_{3}^{\prime
}\left(  k\right)  +u\left(  \sigma\left(  k\right)  ,k\right)  +\frac
{1-\lambda}{\lambda}U\left(  \sigma\left(  k\right)  ,k\right)  =av_{3}\left(
k\right)  +bw_{3}\left(  k\right) \\
\left(  f\left(  k\right)  -\sigma\left(  k\right)  \right)  w_{3}^{\prime
}\left(  k\right)  +u\left(  \sigma\left(  k\right)  ,k\right)  +\frac
{1-\lambda}{\lambda}U\left(  \sigma\left(  k\right)  ,k\right)  =bv_{3}\left(
k\right)  +aw_{3}\left(  k\right)
\end{array}
\label{a28}%
\end{equation}
with the boundary conditions
\begin{equation}%
\begin{array}
[c]{c}%
v_{3}\left(  k_{\infty}\right)  =0\\
w_{3}\left(  k_{\infty}\right)  =0
\end{array}
\label{a29}%
\end{equation}

Obviously, $v_{3}=w_{3}=0$ is a solution. Lemma \ref{lemma.2} below shows that
We need to show that it is the only one, so that $v_{1}=v_{2}$ and
$w_{1}=w_{2}$. This is the desired result.
\end{proof}

\begin{lemma}
\label{lemma.2} If $(v_{3},w_{3})$ is a pair of continuous functions on a
neighborhood $k_{\infty}$, continuously differentiable for $k\neq k_{\infty}$,
and which solve (\ref{a28}) with boundary conditions (\ref{a29}), then
$v_{3}=w_{3}=0$.
\end{lemma}

\begin{proof}
Set $f(k)-\sigma_{1}(k)=\xi(k)$, then $\xi(k)\rightarrow0$ as $k\rightarrow
k_{\infty}$. The system (\ref{a28}) can be rewritten as:
\begin{equation}
{\left(
\begin{array}
[c]{l}%
\xi v_{3}^{\prime}\\
\xi w_{3}^{\prime}%
\end{array}
\right)  }={\left(
\begin{array}
[c]{lr}%
a & b\\
b & a
\end{array}
\right)  }{\left(
\begin{array}
[c]{l}%
v_{3}\\
w_{3}%
\end{array}
\right)  }.\nonumber
\end{equation}

Setting:
\[
{\left(
\begin{array}
[c]{l}%
V\\
W
\end{array}
\right)  }={\left(
\begin{array}
[c]{lr}%
\sqrt{\frac{1}{2}} & \sqrt{\frac{1}{2}}\\
-\sqrt{\frac{1}{2}} & \sqrt{\frac{1}{2}}%
\end{array}
\right)  }{\left(
\begin{array}
[c]{l}%
v_{3}\\
w_{3}%
\end{array}
\right)  },
\]
we have:
\begin{equation}
\xi{\left(
\begin{array}
[c]{l}%
V^{\prime}\\
W^{\prime}%
\end{array}
\right)  }={\left(
\begin{array}
[c]{lr}%
\delta_{1} & 0\\
0 & \delta_{2}%
\end{array}
\right)  }{\left(
\begin{array}
[c]{l}%
V\\
W
\end{array}
\right)  }. \nonumber\label{Matrix2}%
\end{equation}

The first equation yields $V(k)=V(k_{0})\exp\int_{k_{0}}^{k}\frac{\delta_{1}%
}{\xi(u)}du$. Without loss of generality, we can assume that $k_{0}<k_{\infty
}$.

Let $S=\{k\ |\ \xi(k)=0,\ k_{0}\leq k\leq k_{\infty}\}$, where $k_{0}$ is the
initial stock.\ We consider the following two cases:

\textbf{First case}:$\ S=\{k_{\infty}\}$. Then $\xi(k)>0$ for $k\in\lbrack
k_{0},k_{\infty})$. So $\frac{dk}{dt}=\xi(k)>0$, and:
\[
0=\sqrt{\frac{1}{2}}(v_{3}(k_{\infty})+w_{3}(k_{\infty}))=V(k_{\infty
})=V(k_{0})\exp\int_{k_{0}}^{k}\frac{\delta_{1}}{\xi(u)}du
\]

It follows that $V(k)=0$ for all $k$.

\textbf{Second case:}$\ S\backslash\{k_{\infty}\}\neq\emptyset$. Then $S$ is a
closed set and $V(k)=\frac{\xi(k)V^{\prime}(k)}{\delta_{1}}=0$ in $S$. The
complement of $S$ is a countable union of disjoint intervals, and $\xi$
vanishes at the endpoints of each interval.\ Arguing as above it follows that
$V\left(  k\right)  $ vanishes for $k_{0}\leq k\leq k_{\infty}$.

The same argument holds for $W\left(  k\right)  $. This concludes the proof.
\end{proof}

\subsubsection{The Euler-type equations}

The Euler-type equations (\ref{a15}) to (\ref{a19}) can be reduced to a
non-autonomous system of four first-order ODEs. It is better to work directly
on (\ref{35}) and (\ref{36}) to get an autonomous system. Proceeding as in the
end of Section 2, we differentiate (\ref{35}) and (\ref{36}) w.r.t. $k$, 
and notice that $(f-\sigma)V_{\lambda}''=\frac{dk}{dt}\frac{V_{\lambda}'}{dk}=\frac{d}{dt}V_{\lambda}'$,
we get:%
\begin{align*}
\frac{\delta_{1}+\delta_{2}}{2}V_{\lambda}^{\prime}+\frac{\delta_{1}%
-\delta_{2}}{2}W_{\lambda}^{\prime}  &  =\frac{d}{dt}V_{\lambda}^{\prime
}+(f^{\prime}-\sigma^{\prime})V_{\lambda}^{\prime}+(u_{1}^{\prime}%
+\frac{1-\lambda}{\lambda}U_{1}^{\prime})\sigma^{\prime}+u_{2}^{\prime}%
+\frac{1-\lambda}{\lambda}U_{2}^{\prime}\\
\frac{\delta_{1}-\delta_{2}}{2}V_{\lambda}^{\prime}+\frac{\delta_{1}%
+\delta_{2}}{2}W_{\lambda}^{\prime}  &  =\frac{d}{dt}W_{\lambda}^{\prime
}+(f^{\prime}-\sigma^{\prime})W_{\lambda}^{\prime}+(u_{1}^{\prime}%
-\frac{1-\lambda}{\lambda}U_{1}^{\prime})\sigma^{\prime}+u_{2}^{\prime}%
-\frac{1-\lambda}{\lambda}U_{2}^{\prime}%
\end{align*}
where $c=\varphi\left(  V_{\lambda}^{\prime}\left(  k\right)  ,k\right)  $ and
$\sigma^{\prime}$ is given by (\ref{sigma'.of.kc}). Using formula (\ref{39}),
and setting $w\left(  t\right)  :=W_{\lambda}^{\prime}\left(  k\left(
t\right)  \right)  $, this becomes:%
\begin{align}
\frac{\delta_{1}+\delta_{2}}{2}\left(u_{1}^{\prime
}+\frac{1-\lambda}{\lambda}U_{1}^{\prime}\right)  +\frac{\delta_{1}%
-\delta_{2}}{2}w  &  =\frac{d}{dt}\left(u_{1}^{\prime
}+\frac{1-\lambda}{\lambda}U_{1}^{\prime}\right)+(f'-\sigma')\left(u_{1}^{\prime
}+\frac{1-\lambda}{\lambda}U_{1}^{\prime}\right)
\nonumber\\
&\quad+\left(u_{1}^{\prime
}+\frac{1-\lambda}{\lambda}U_{1}^{\prime}\right)\sigma'  +u_{2}^{\prime
}+\frac{1-\lambda}{\lambda}U_{2}^{\prime},
\nonumber\\
\frac{\delta_{1}-\delta_{2}}{2}\left(u_{1}^{\prime
}+\frac{1-\lambda}{\lambda}U_{1}^{\prime}\right)  +\frac{\delta_{1}+\delta_{2}}%
{2}w  &  =\frac{dw}{dt}+(f'-\sigma')w 
\nonumber\\
&  \quad+\left(u_{1}^{\prime}-\frac{1-\lambda}{\lambda}U_{1}^{\prime}\right)\sigma^{\prime}
+u_{2}^{\prime}-\frac{1-\lambda}{\lambda}U_{2}^{\prime},\nonumber
\end{align}
and hence
\begin{align}
(\frac{\delta_{1}+\delta_{2}}{2}-f^{\prime})\left(  \lambda u_{1}^{\prime
}+\left(  1-\lambda\right)  U_{1}^{\prime}\right)  +\frac{\delta_{1}%
-\delta_{2}}{2}\lambda w  &  =\frac{d}{dt}\left(  \lambda u_{1}^{\prime
}+\left(  1-\lambda\right)  U_{1}^{\prime}\right)  +\lambda u_{2}^{\prime
}+\left(  1-\lambda\right)  U_{2}^{\prime},\label{42}\\
\frac{\delta_{1}-\delta_{2}}{2}\left(  \lambda u_{1}^{\prime}+\left(
1-\lambda\right)  U_{1}^{\prime}\right)  +(\frac{\delta_{1}+\delta_{2}}%
{2}-f^{\prime})\lambda w  &  =\lambda\frac{dw}{dt} +(-\lambda w+\lambda
u_{1}^{\prime}-\left(  1-\lambda\right)  U_{1}^{\prime})\sigma^{\prime
}\nonumber\\
&  \quad+\lambda u_{2}^{\prime}-\left(  1-\lambda\right)  U_{2}^{\prime}.
\label{43}%
\end{align}
to which we should add:%
\begin{equation}
\frac{dk}{dt}=f\left(  k\right)  -c \label{44}%
\end{equation}

The system (\ref{42}) to (\ref{44}) is a system of three first-order ODEs for
the unknown functions $k\left(  t\right)  ,c\left(  t\right)  $ and $w\left(
t\right)  $. Note that for $\lambda=0$, (\ref{42}) and (\ref{43}) reduce to
two copies of the usual Euler-Lagrange equation. For $\lambda=1$, taking
$w\left(  t\right)  =u_{1}^{\prime}\left(  c\left(  t\right)  ,k\left(
t\right)  \right)  $ gives us again two copies of the same equation.

\subsubsection{The control theory approach}

Equation (\ref{K}) is the HJB equation for the following:
\begin{align}
&  \max_{c\left(  .\right)  }\int_{0}^{\infty}e^{-rt}\left[  u\left(  c\left(
t\right)  ,k\left(  t\right)  \right)  +\frac{1-\lambda}{\lambda}U\left(
c\left(  t\right)  ,k\left(  t\right)  \right)  -K\left(  k\left(  t\right)
\right)  \right]  dt\label{45}\\
&  \frac{dk}{dt}=f\left(  k\right)  -c\left(  t\right)  ,\ \ k\left(
0\right)  =k_{0} \label{46}%
\end{align}
where we seek a feedback $\sigma\left(  k\right)  $ such that:%
\begin{equation}
K(k_{0})=\left\{  (\delta-r)\int_{0}^{\infty}e^{-\delta t}u(\sigma\left(
k\left(  t\right)  \right)  ,k\left(  t\right)  )dt\ |\
\begin{array}
[c]{c}%
\frac{dk}{dt}=f\left(  k\right)  -\sigma\left(  k\right) \\
k\left(  0\right)  =k_{0}%
\end{array}
\right\}  \label{47}%
\end{equation}

Solving (\ref{45}), (\ref{46}) under the constraint (\ref{47}) is a
fixed-point problem for the feedback $c=\sigma\left(  k\right)  $.

\subsection{Solving the boundary-value problem}

Define a function $\overline{g}_{\lambda}(k)$ by:%
\begin{equation}
\overline{g}_{\lambda}(k):=\frac{\lambda\delta_{1}{u_{1}^{\prime}}\left(
f\left(  k\right)  ,k\right)  +(1-\lambda)\delta_{2}U_{1}^{\prime}\left(
f\left(  k\right)  ,k\right)  }{\lambda u_{1}^{\prime}\left(  f\left(
k\right)  ,k\right)  +(1-\lambda)U_{1}^{\prime}\left(  f\left(  k\right)
,k\right)  }-\frac{\lambda u_{2}^{\prime}\left(  f\left(  k\right)  ,k\right)
+(1-\lambda)U_{2}^{\prime}\left(  f\left(  k\right)  ,k\right)  }{\lambda
u_{1}^{\prime}\left(  f\left(  k\right)  ,k\right)  +(1-\lambda)U_{1}^{\prime
}\left(  f\left(  k\right)  ,k\right)  } \label{bar.g}%
\end{equation}

\begin{theorem}
\label{find.the.strategies} Assume $\overline{g}_{\lambda}(k)\neq0$. If there
is some $k_{\infty}$ such that
\begin{equation}
f^{\prime}(k_{\infty})\ne\overline{g}_{\lambda}(k_{\infty}), \label{53}%
\end{equation}
then the equations (\ref{35}) and (\ref{36}) with the boundary conditions
(\ref{37}) and (\ref{38}) have a solution $\left(  V,W\right)  $ near the
point $k_{\infty}$ with $V$ of class $C^{2}$ and $W$ of class $C^{1.}$
\end{theorem}

\noindent\textbf{Proof.} We adapt the argument in \cite{EkL}\thinspace\ to the
present situation. We note that the boundary-value problem (\ref{35}),
(\ref{36}), (\ref{37}), (\ref{38}) cannot be reduced to a standard
initial-value problem for the pair $\left(  V_{\lambda},W_{\lambda}\right)  $.
To see that, rewrite equation (\ref{35}) as follows:%
\[
\left(  f\left(  k\right)  -\varphi\left(  V_{\lambda}^{\prime},k\right)
\right)  V_{\lambda}^{\prime}+u(\varphi\left(  V_{\lambda}^{\prime},k\right)
,k)+\frac{1-\lambda}{\lambda}U(\varphi\left(  V_{\lambda}^{\prime},k\right)
,k)=aV_{\lambda}+bW_{\lambda}.
\]

Since the function $c\rightarrow u(c,k)+\frac{1-\lambda}{\lambda}U(c,k)$ is
concave, the function:%
\[
c\rightarrow u(c,k)+\frac{1-\lambda}{\lambda}U(c,k)-cx
\]
attains its maximum at the point $c=\varphi_{\lambda}\left(  x,k\right)  $
defined by (\ref{40}). We set:%
\begin{equation}
u_{\lambda}^{\ast}\left(  x,k\right)  =\max_{c}\left\{  u(c,k)+\frac
{1-\lambda}{\lambda}U(c,k)-cx\right\}  =u\left(  \varphi_{\lambda}\left(
x,k\right)  ,k\right)  +\frac{1-\lambda}{\lambda}U\left(  \varphi_{\lambda
}\left(  x,k\right)  ,k\right)  -\varphi_{\lambda}\left(  x,k\right)  x
\label{52}%
\end{equation}

The function $x\rightarrow u_{\lambda}^{\ast}\left(  x,k\right)  $ is convex,
and the equation (\ref{35}) becomes:%
\begin{equation}
f\left(  k\right)  V_{\lambda}^{\prime}+u_{\lambda}^{\ast}\left(  V_{\lambda
}^{\prime},k\right)  =aV_{\lambda}+bW_{\lambda}. \label{41}%
\end{equation}

This is an equation for $V_{\lambda}^{\prime}$. From the basic duality results
in convex analysis (see for instance \cite{ET}), we find that:%
\[
\min_{y}\left\{  f\left(  k\right)  y+u_{\lambda}^{\ast}\left(  y,k\right)
\right\}  =u(f\left(  k\right)  ,k)+\frac{1-\lambda}{\lambda}U(f\left(
k\right)  ,k).
\]

Note that $f(k)y+u_{\lambda}^{\ast}\left(  y,k\right)  $ is convex in $y$ with
minimal value $u(f(k),k)+\frac{1-\lambda}{\lambda}U(f(k),k)$. Then (\ref{41}),
considered as an equation for $V_{\lambda}^{\prime}$, has two solutions if:%
\[
aV_{\lambda}\left(  k\right)  +bW_{\lambda}\left(  k\right)  >u(f\left(
k\right)  ,k)+\frac{1-\lambda}{\lambda}U(f\left(  k\right)  ,k),
\]
and no solutions if the inverse inequality holds.\ If we have equality:%
\[
aV_{\lambda}\left(  k\right)  +bW_{\lambda}\left(  k\right)  =u(f\left(
k\right)  ,k)+\frac{1-\lambda}{\lambda}U(f\left(  k\right)  ,k)
\]
then equation (\ref{41}) has precisely one solution, namely:%
\[
V_{\lambda}^{\prime}=u_{1}^{\prime}(f\left(  k\right)  ,k)+\frac{1-\lambda
}{\lambda}U_{1}^{\prime}(f\left(  k\right)  ,k)
\]

At the point $k_{\infty}$, with the values (\ref{37}), (\ref{38}), we find
that:
\[
aV_{\lambda}\left(  k_{\infty}\right)  +bW_{\lambda}\left(  k\right)
=u\left(  f\left(  k_{\infty}\right)  ,k_{\infty}\right)  +\frac{1-\lambda
}{\lambda}U\left(  f\left(  k_{\infty}\right)  ,k_{\infty}\right)  ,
\]
so that we are exactly on the boundary case. Equation (\ref{41}) has precisely
one solution, by (\ref{Vr'}) below, which satisfies
\[
V_{\lambda}^{\prime}\left(  k_{\infty}\right)  =u_{1}^{\prime}(f\left(
k_{\infty}\right)  ,k_{\infty})+\frac{1-\lambda}{\lambda}U_{1}^{\prime
}(f\left(  k_{\infty}\right)  ,k_{\infty}),
\]
but it is degenerate:%
\[
\frac{\partial}{\partial y}\left[  f\left(  k_{\infty}\right)  y+u_{\lambda
}^{\ast}\left(  y,k_{\infty}\right)  \right]  |_{y=V_{\lambda}^{\prime}\left(
k_{\infty}\right)  }=0,
\]
so that equation (\ref{41}) cannot be written in the form $V_{\lambda}%
^{\prime}=\psi\left(  k,V_{\lambda}\right)  $. It is an implicit differential
equation, and special techniques are needed to solve it.\ The difficulty is
further increased by the fact that we need only $V$ (but not $W$) to be
$C^{2}$.

The proof of Theorem \ref{find.the.strategies} proceeds in several steps.

\textbf{Step 1: Changing the unknown functions from }$\left(  V_{\lambda
}\left(  k\right)  ,W_{\lambda}\left(  k\right)  \right)  $ \textbf{to}
$\left(  V_{\lambda}\left(  k\right)  ,\mu\left(  k\right)  \right)  $

Using (\ref{52}), we find that the equation (\ref{35}) is equivalent to
\[
F\left(  V^{\prime}\left(  k\right)  ,k\right)  =\mu(k)
\]
with:%
\begin{align*}
F\left(  x,k\right)   &  :=u_{\lambda}^{\ast}\left(  x,k\right)  +xf\left(
k\right)  -u(f(k),k)-\frac{1-\lambda}{\lambda}U(f(k),k)\\
\mu(k)  &  :=aV_{\lambda}(k)+bW_{\lambda}(k)-u(f(k),k)-\frac{1-\lambda
}{\lambda}U(f(k),k)
\end{align*}

Note that $x\rightarrow F\left(  x,k\right)  $ is a linear perturbation of the
convex function $x\rightarrow u_{\lambda}^{\ast}\left(  x,k\right)  $. To
shorten the notations, let us set:%
\begin{equation}
\label{y.star}y^{\ast}\left(  k\right)  :=u_{1}^{\prime}(f\left(  k\right)
,k)+\frac{1-\lambda}{\lambda}U_{1}^{\prime}(f\left(  k\right)  ,k)
\end{equation}

As we pointed out, the equation
\[
F(x+y^{\ast}\left(  k\right)  ,k)=\mu
\]
in the variable $x$ has two solutions $x_{-}\left(  k,\mu\right)
<0<x_{+}\left(  k,\mu\right)  $ for $\mu(k)>0$, none for $\mu(k)<0$ and a
single solution $x=0$ for $\mu(k)=0$.

From the equation (\ref{36}), we have
\begin{equation}
W_{\lambda}^{\prime}=V_{\lambda}^{\prime}\frac{bV_{\lambda}+aW_{\lambda
}-u\left(  \sigma_{\lambda}\left(  k\right)  ,k\right)  +\frac{1-\lambda
}{\lambda}U\left(  \sigma_{\lambda}\left(  k\right)  ,k\right)  }{aV_{\lambda
}+bW_{\lambda}-u\left(  \sigma_{\lambda}\left(  k\right)  ,k\right)
-\frac{1-\lambda}{\lambda}U\left(  \sigma_{\lambda}\left(  k\right)
,k\right)  } \nonumber\label{w'}%
\end{equation}

Differentiating $\mu(k)$ with respect to $k$ yields
\begin{align}
\frac{d\mu}{dk}  &  =V_{\lambda}^{\prime}\frac{(a^{2}+b^{2})V_{\lambda
}+2abW_{\lambda}-\delta_{1}{u}\left(  \sigma_{\lambda}\left(  k\right)
,k\right)  -\frac{(1-\lambda)\delta_{2}}{\lambda}U\left(  \sigma_{\lambda
}\left(  k\right)  ,k\right)  }{aV_{\lambda}+bW_{\lambda}-u\left(
\sigma_{\lambda}\left(  k\right)  ,k\right)  -\frac{1-\lambda}{\lambda
}U\left(  \sigma_{\lambda}\left(  k\right)  ,k\right)  }\nonumber\\
&  -y^{\ast}f^{\prime}-u_{2}^{\prime}\left(  f\left(  k\right)  ,k\right)
-\frac{1-\lambda}{\lambda}U_{2}^{\prime}\left(  f\left(  k\right)  ,k\right)
\nonumber
\end{align}

We now take $(V_{\lambda}(k),\mu(k))$ as our new unknown functions. They
satisfy the equations:
\begin{align*}
\frac{dV}{dk}  &  =y^{\ast}\left(  k\right)  +x\left(  k,\mu\left(  k\right)
\right)  ,\\
\frac{d\mu}{dk}  &  =\left(  y^{\ast}+x\left(  k,\mu\left(  k\right)  \right)
\right)  \frac{(a^{2}+b^{2})V_{\lambda}+2abW_{\lambda}-\delta_{1}{u}\left(
\sigma_{\lambda}\left(  k\right)  ,k\right)  -\frac{(1-\lambda)\delta_{2}%
}{\lambda}U\left(  \sigma_{\lambda}\left(  k\right)  ,k\right)  }{aV_{\lambda
}+bW_{\lambda}-u\left(  \sigma_{\lambda}\left(  k\right)  ,k\right)
-\frac{1-\lambda}{\lambda}U\left(  \sigma_{\lambda}\left(  k\right)
,k\right)  }\\
&  \quad-y^{\ast}f^{\prime}-u_{2}^{\prime}\left(  f\left(  k\right)
,k\right)  -\frac{1-\lambda}{\lambda}U_{2}^{\prime}\left(  f\left(  k\right)
,k\right)  .
\end{align*}

In fact, according to which determination $\left(  x\left(  k,\mu\left(
k\right)  \right)  \right)  $ is chosen, $x_{+}\left(  k,\mu\left(  k\right)
\right)  $ or $x_{-}\left(  k,\mu\left(  k\right)  \right)  $, these equations
define two distinct dynamical systems on the region $\mu>0$.

\textbf{Step 3: Taking }$x$ \textbf{instead of} $k$ \textbf{as the independent
variable}.

To get rid of the indetermination, we pick $x$ instead of $k$ as the
independent variable. We shorten our notation by setting $x(k)=x(\mu(k),k)$.
We get:
\begin{equation}
\frac{dk}{dx}=\frac{f(k)-\varphi_{\lambda}(y^{\ast}+x,k)}{D(x,k,V_{\lambda
},W_{\lambda})}\left[  aV_{\lambda}+bW_{\lambda}-u(\varphi_{\lambda}(y^{\ast
}+x,k),k)-\frac{1-\lambda}{\lambda}U(\varphi_{\lambda}(y^{\ast}%
+x,k),k)\right]  \label{dk.dx}%
\end{equation}
where (here $\varphi_{\lambda}$ stands for $\varphi_{\lambda}(y^{\ast}+x,k)$)
\begin{align*}
D(x,k,V_{\lambda},W_{\lambda})  &  =(y^{\ast}+x)\left[  (a^{2}+b^{2}%
)V_{\lambda}+2abW_{\lambda}-\delta_{1}{u}\left(  \sigma_{\lambda}\left(
k\right)  ,k\right)  -\frac{(1-\lambda)\delta_{2}}{\lambda}U\left(
\sigma_{\lambda}\left(  k\right)  ,k\right)  \right] \\
&  \quad-A(aV_{\lambda}+bW_{\lambda}-u(\varphi_{\lambda},k)-\frac{1-\lambda
}{\lambda}U(\varphi_{\lambda},k))\\
A(x,k,V_{\lambda},W_{\lambda})  &  =(f-\varphi_{\lambda})\left[  \left(
u_{11}^{\prime\prime}\left(  f\left(  k\right)  ,k\right)  +\frac{1-\lambda
}{\lambda}U_{11}^{\prime\prime}\left(  f\left(  k\right)  ,k\right)  \right)
f^{\prime}+u_{12}^{\prime\prime}\left(  f\left(  k\right)  ,k\right)  \right.
\\
&  \quad\left.  +\frac{1-\lambda}{\lambda}U_{12}^{\prime\prime}\left(
f\left(  k\right)  ,k\right)  \right]  \ +\left(  y^{\ast}+x\right)
f^{\prime}(k)+u_{2}^{\prime}(\varphi_{\lambda},k)+\frac{1-\lambda}{\lambda
}U_{2}^{\prime}(\varphi_{\lambda},k)
\end{align*}

Further more, we have
\begin{equation}
\frac{dV_{\lambda}}{dx}=\frac{dV_{\lambda}}{dk}\frac{dk}{dx}=(y^{\ast}%
+x)\frac{f(k)-\varphi_{\lambda}}{D(x,k,V_{\lambda},W_{\lambda})}[aV_{\lambda
}+bW_{\lambda}-u(\varphi_{\lambda},k)-\frac{1-\lambda}{\lambda}U(\varphi
_{\lambda},k)]. \label{dvdx}%
\end{equation}

\textbf{Step 4: Rescaling the time}

We introduce a new variable $s$ such that $D(x,k,V_{\lambda},W_{\lambda
})ds=dx$, then the system becomes
\[%
\begin{array}
[c]{rcl}%
\frac{dx}{ds} & = & D(x,k,V_{\lambda},W_{\lambda})\\
\frac{dk}{dx} & = & \left(  f(k)-\varphi_{\lambda}\right)  \left[
aV_{\lambda}+bW_{\lambda}-u(\varphi_{\lambda}(y^{\ast}+x,k),k)-\frac
{1-\lambda}{\lambda}U(\varphi_{\lambda}(y^{\ast}+x,k),k)\right] \\
\frac{dV_{\lambda}}{dx} & = & (y^{\ast}+x)\left(  f(k)-\varphi_{\lambda
}(y^{\ast}+x,k)\right)  [aV_{\lambda}+bW_{\lambda}-u(\varphi_{\lambda
},k)-\frac{1-\lambda}{\lambda}U(\varphi_{\lambda},k)]
\end{array}
\]

We now eliminate $W_{\lambda}$, to get an equation in $\left(  x,k,V\right)  $
only. From the equation of $\mu\left(  k\right)  =F\left(  x,k\right)  $, we
have
\[
bW_{\lambda}(k)=F(y^{\ast}(k)+x,k)+u(f(k),k)+\frac{1-\lambda}{\lambda
}U(f(k),k)-aV_{\lambda}(k)
\]

The dynamics of $\left(  x\left(  s\right)  ,k\left(  s\right)  ,V_{\lambda
}\left(  s\right)  \right)  $ are given by:
\begin{equation}%
\begin{array}
[c]{rcl}%
\frac{dx}{ds} & = & \tilde{D}(x,k,V_{\lambda})\\
\frac{dk}{dx} & = & \left(  f(k)-\varphi_{\lambda}\right)  \left[  F(y^{\ast
}(k)+x,k)+u(f(k),k)+\frac{1-\lambda}{\lambda}U(f(k),k)\right. \\
&  & \left.  -u(\varphi_{\lambda}(y^{\ast}+x,k),k)-\frac{1-\lambda}{\lambda
}U(\varphi_{\lambda}(y^{\ast}+x,k),k)\right] \\
\frac{dV_{\lambda}}{dx} & = & (y^{\ast}+x)\left(  f(k)-\varphi_{\lambda
}(y^{\ast}+x,k)\right)  [F(y^{\ast}(k)+x,k)+u(f(k),k)+\frac{1-\lambda}%
{\lambda}U(f(k),k)\\
&  & -u(\varphi_{\lambda}(y^{\ast}+x,k),k)-\frac{1-\lambda}{\lambda}%
U(\varphi_{\lambda}(y^{\ast}+x,k),k)]
\end{array}
\label{sx}%
\end{equation}
where $\tilde{D}(x,k,V_{\lambda})=$ $D(x,k,V_{\lambda},W_{\lambda})$. For
$\tilde{D}$ to be $C^{2}$, we need $f$ to be $C^{3}$ and $u,U$ to be $C^{4}$.

\textbf{Step 5: Linearizing the system}

As we already noted, if $k=k_{\infty}$, then $\mu(k_{\infty})=0$ and
$F(y^{\ast}(k_{\infty})+x,k_{\infty})=\mu(k_{\infty})=0$ has only one solution
$x=0$. Set $v_{\infty}=V_{\lambda}(k_{\infty})$. We consider the system near
the point $(x,k,V_{\lambda})=(0,k_{\infty},v_{\infty})$. For simplicity, we
write
\[
y_{\infty}=y^{\ast}(k_{\infty})=u_{1}^{\prime}(f(k_{\infty}),k_{\infty}%
)+\frac{1-\lambda}{\lambda}U_{1}^{\prime}(f(k_{\infty}),k_{\infty})
\]

Computing the linearized system at $(0,k_{\infty},v_{\infty})$, we find:
\begin{equation}
\frac{d}{ds}{\left(
\begin{array}
[c]{l}%
x\\
k-k_{\infty}\\
V_{\lambda}-v_{\infty}%
\end{array}
\right)  }={\left(
\begin{array}
[c]{lcr}%
a_{\infty} & b_{\infty} & c_{\infty}\\
0 & 0 & 0\\
0 & 0 & 0
\end{array}
\right)  }{\left(
\begin{array}
[c]{l}%
x\\
k-k_{\infty}\\
V_{\lambda}-v_{\infty}%
\end{array}
\right)  }. \nonumber\label{linearized. system.1}%
\end{equation}
where:%
\begin{align*}
a_{\infty}  &  =y^{\ast}(k_{\infty})^{2}\frac{\partial\varphi_{\lambda}%
}{\partial y}\left(  y^{\ast}(k_{\infty}),k_{\infty}\right)  (f^{\prime
}(k_{\infty})-\overline{g}_{\lambda}(k_{\infty}))\\
\overline{g}_{\lambda}(k_{\infty})  &  =\frac{\lambda\delta_{1}{u_{1}^{\prime
}}\left(  f\left(  k_{\infty}\right)  ,k_{\infty}\right)  +(1-\lambda
)\delta_{2}U_{1}^{\prime}\left(  f\left(  k_{\infty}\right)  ,k_{\infty
}\right)  } {\lambda u_{1}^{\prime}\left(  f\left(  k_{\infty}\right)
,k_{\infty}\right)  +(1-\lambda)U_{1}^{\prime}\left(  f\left(  k_{\infty
}\right)  ,k_{\infty}\right)  }\\
&  \quad-\frac{\lambda u_{2}^{\prime}\left(  f\left(  k_{\infty}\right)
,k_{\infty}\right)  +(1-\lambda)U_{2}^{\prime}\left(  f\left(  k_{\infty
}\right)  ,k_{\infty}\right)  } {\lambda u_{1}^{\prime}\left(  f\left(
k_{\infty}\right)  ,k_{\infty}\right)  +(1-\lambda)U_{1}^{\prime}\left(
f\left(  k_{\infty}\right)  ,k_{\infty}\right)  }%
\end{align*}

By the concavity assumptions, we find:
\begin{align*}
&  y_{\infty}=u_{1\infty}^{\prime}+\frac{1-\lambda}{\lambda}U_{1\infty
}^{\prime}=u_{1}^{\prime}(f(k_{\infty}),k_{\infty})+\frac{1-\lambda}{\lambda
}U_{1}^{\prime}(f(k_{\infty}),k_{\infty})>0,\\
&  u_{11}^{\prime}(f(k_{\infty}),k_{\infty})+\frac{1-\lambda}{\lambda}%
U_{11}^{\prime}(f(k_{\infty}),k_{\infty})<0\\
\frac{\partial\varphi_{\lambda}}{\partial y}\left(  y^{\ast}(k_{\infty
}),k_{\infty}\right)   &  =\left(  u_{1}^{\prime\prime}(f(k_{\infty
}),k_{\infty})+\frac{\lambda}{1-\lambda}U_{1}^{\prime\prime}(f(k_{\infty
}),k_{\infty})\right)  ^{-1}<0
\end{align*}

For any $k_{\infty}$ satisfying (\ref{53}) with $\overline{g}_{\lambda
}(k_{\infty})-f^{\prime}(k_{\infty})\neq0$, we then have $a_{\infty}\neq0$.
Introducing:
\begin{equation}
\tilde{x}=x+\frac{b_{\infty}}{a_{\infty}}(k-k_{\infty})+\frac{c_{\infty}%
}{a_{\infty}}(V_{\lambda}-v_{\infty}), \label{tilde.x}%
\end{equation}
we transform the system (\ref{sx}) for $\left(  \tilde{x}\left(  s\right)
,k\left(  s\right)  ,V_{\lambda}\left(  s\right)  \right)  $. The
linearization at the origin is given by:
\begin{equation}
\frac{d}{ds}{\left(
\begin{array}
[c]{l}%
\tilde{x}\\
k-k_{\infty}\\
V_{\lambda}-v_{\infty}%
\end{array}
\right)  }={\left(
\begin{array}
[c]{lcr}%
-y_{\infty}^{2}\varphi_{1{\infty}}^{\prime}(\overline{g}_{\lambda}(k_{\infty
})-f^{\prime}(k_{\infty})) & 0 & 0\\
0 & 0 & 0\\
0 & 0 & 0
\end{array}
\right)  }{\left(
\begin{array}
[c]{l}%
\tilde{x}\\
k-k_{\infty}\\
V_{\lambda}-v_{\infty}%
\end{array}
\right)  }. \nonumber\label{linearized. system.2}%
\end{equation}

\textbf{Step 6: Applying the center manifold theorem}

By the center manifold theorem (cf. Theorem 1 of \cite{Car1}), there exist an
$\epsilon>0$, and a map $h(k,V_{\lambda})$, defined in a neighborhood
$\mathcal{O}=(k_{\infty}-\epsilon,k_{\infty}+\epsilon)\times(v_{\infty
}-\epsilon,v_{\infty}+\epsilon)$ of $(k_{\infty},v_{\infty})$ such that
\[
h(k_{\infty},v_{\infty})=0,\quad\frac{\partial{h}}{\partial{k}}(k_{\infty
},v_{\infty})=0,\quad\frac{\partial{h}}{\partial{V_{\lambda}}}(k_{\infty
},v_{\infty},0)=0
\]
and the manifold $\mathcal{M}$ defined by
\[
\mathcal{M}=\left\{  \left(  h(k,V_{\lambda})-\frac{b_{\infty}}{a_{\infty}%
}(k-k_{\infty})-\frac{c_{\infty}}{a_{\infty}}(V_{\lambda}-v_{\infty
}),k,V_{\lambda}\right)  \ |\ (k,V_{\lambda})\in\mathcal{O}\right\}  ,
\]
is invariant under the flow associated to the system (\ref{sx}). The map $h$
and the central manifold $\mathcal{M}$ are $C^{2}$, and $\mathcal{M}$ is
two-dimensional and tangent to the critical plane defined by $\tilde{x}=0$.

If $k=k_{\infty}$ and $V_{\lambda}=v_{\infty}$, then $x=0$, and $h(k_{\infty
},v_{\infty})=\tilde{x}(k_{\infty},v_{\infty})=x(k_{\infty},v_{\infty})=0$

We are interested in the solutions which lie on the central manifold
$\mathcal{M}$. Writing:
\[
x=h(k,V_{\lambda})-\frac{b_{\infty}}{a_{\infty}}(k-k_{\infty})-\frac
{c_{\infty}}{a_{\infty}}(V_{\lambda}-v_{\infty})
\]
in the equation $\frac{dV}{dk}=y^{\ast}\left(  k\right)  +x$, we get
\begin{equation}
\left\{
\begin{array}
[c]{lcl}%
\frac{dV_{\lambda}}{dk} & = & u_{1}^{\prime}\left(  f\left(  k\right)
,k\right)  +\frac{1-\lambda}{\lambda}U_{1}^{\prime}\left(  f\left(  k\right)
,k\right)  +h\left(  k,V_{\lambda}\right)  -\frac{b_{\infty}}{a_{\infty}%
}\left(  k-k_{\infty}\right)  -\frac{c_{\infty}}{a_{\infty}}\left(
V_{\lambda}-v_{\infty}\right)  ,\\
V_{\lambda}\left(  k_{\infty}\right)  & = & v_{\infty} \label{vk}%
\end{array}
\right.
\end{equation}
which can be viewed as eliminating the variable $s$ from the second and third
equations of the system (\ref{sx}). Since $a_{\infty}\neq0$ and the right hand
side of the first equation of (\ref{vk}) is continuously differentiable in
$\mathcal{O}_{1}=(k_{\infty}-\epsilon,k_{\infty}+\epsilon)$, therefore, is
locally Lipschiz continuous. By $\frac{d{V_{\lambda}}}{d{k}}|_{k=k_{\infty}%
}=y_{\infty}\neq0$ which follows from (\ref{y.star}) and the first equation of
(\ref{vk}), the nonconstant solution of this initial-value problem exist in
$\mathcal{O}_{1}$ which we denote by $V_{\lambda}(k)=\zeta(k)$ ,where
$\zeta(k_{\infty})=v_{\infty}$ and $\zeta\in{C}^{2}(\mathcal{O}_{1})$ if
$h\in{C}^{2}(\mathcal{O})$. Substituting $V_{\lambda}(k)=\zeta(k)$ into
$x=h(k,V_{\lambda})-\frac{b_{\infty}}{a_{\infty}}(k-k_{\infty})-\frac
{c_{\infty}}{a_{\infty}}(V_{\lambda}-v_{\infty})$ yields
\[
x(k)=h(k,\zeta(k))-\frac{b_{\infty}}{a_{\infty}}(k-k_{\infty})-\frac
{c_{\infty}}{a_{\infty}}(\zeta(k)-v_{\infty}).
\]
Finally,
\[
\mu(k)=F(u_{1}^{\prime}(f(k),k)+\frac{1-\lambda}{\lambda}U_{1}^{\prime
}(f(k),k)+x(k),k)
\]
and
\[
W_{\lambda}(k)=\frac{1}{b}(\mu(k)+u(f(k),k)+\frac{1-\lambda}{\lambda
}U(f(k),k)-a\zeta(k)
\]
is also $C^{2}$, so we have found a $C^{2}$ solution of the system
(\ref{35}),(\ref{36}) with boundary conditions (\ref{37}),(\ref{38}). \newline

\subsection{The existence of equilibrium strategies}

Introduce another function $\underline{g}_{\lambda}(k)$ defined by:%

\begin{equation}
\underline{g}_{\lambda}(k)=\frac{\lambda{u}_{1}^{\prime}\left(  f\left(
k\right)  ,k\right)  +(1-\lambda)U_{1}^{\prime}\left(  f\left(  k\right)
,k\right)  }{\frac{\lambda}{\delta_{1}}u_{1}^{\prime}\left(  f\left(
k\right)  ,k\right)  +\frac{1-\lambda}{\delta_{2}}U_{1}^{\prime}\left(
f\left(  k\right)  ,k\right)  }-\frac{\lambda\delta_{2}u_{2}^{\prime}\left(
f\left(  k\right)  ,k\right)  +\left(  1-\lambda\right)  \delta_{1}%
U_{2}^{\prime}\left(  f\left(  k\right)  ,k\right)  }{\lambda\delta_{2}%
u_{1}^{\prime}\left(  f\left(  k\right)  ,k\right)  +\left(  1-\lambda\right)
\delta_{1}U_{1}^{\prime}\left(  f\left(  k\right)  ,k\right)  }
\label{under.g}%
\end{equation}

\begin{theorem}
\label{Th1}Suppose $f^{\prime}\left(  k_{\infty}\right)  $ lies between
$\underline{g}_{\lambda}(k_{\infty})$ and $\bar{g}_{\lambda}(k_{\infty})$.
Then there exists an equilibrium strategy converging to $k_{\infty}$
\end{theorem}

Comparing (\ref{bar.g}) and (\ref{under.g}), we find that:%

\begin{equation}
\overline{g}_{\lambda}(k)-\underline{g}_{\lambda}(k)=\lambda(1-\lambda
)(\delta_{1}-\delta_{2})\frac{\left(  \delta_{1}-\delta_{2}\right)
u_{1}^{\prime}U_{1}^{\prime}-u_{2}^{\prime}U_{1}^{\prime}+u_{1}^{\prime}%
U_{2}^{\prime}}{(\lambda u_{1}^{\prime}+(1-\lambda)U_{1}^{\prime}%
)(\lambda\delta_{2}u_{1}^{\prime}+\left(  1-\lambda\right)  \delta_{1}%
U_{1}^{\prime})}\nonumber
\end{equation}

So the sign of $\overline{g}_{\lambda}(k)-\underline{g}_{\lambda}(k)$ is the
sign of $\left(  \delta_{1}-\delta_{2}\right)  u_{1}^{\prime}U_{1}^{\prime
}-u_{2}^{\prime}U_{1}^{\prime}+u_{1}^{\prime}U_{2}^{\prime}$:%
\[%
\begin{array}
[c]{c}%
\delta_{1}-\delta_{2}>\frac{u_{2}^{\prime}}{u_{1}^{\prime}}-\frac
{U_{2}^{\prime}}{U_{1}^{\prime}}\Longrightarrow\overline{g}_{\lambda
}(k)>\underline{g}_{\lambda}(k)\\
\delta_{1}-\delta_{2}<\frac{u_{2}^{\prime}}{u_{1}^{\prime}}-\frac
{U_{2}^{\prime}}{U_{1}^{\prime}}\Longrightarrow\overline{g}_{\lambda
}(k)<\underline{g}_{\lambda}(k)
\end{array}
\]

Let us give some examples:

\textbf{Example 1:\ }$u\left(  c\right)  =U\left(  c\right)  $

In this case, $u=U$ and they do not depend on $k$. We get:%
\begin{align*}
\overline{g}_{\lambda}(k)  &  =\lambda\delta_{1}+(1-\lambda)\delta_{2}\\
\underline{g}_{\lambda}(k)  &  =\frac{1}{\frac{\lambda}{\delta_{1}}%
+\frac{\left(  1-\lambda\right)  }{\delta_{2}}}%
\end{align*}

These formulas do not depend on the the utility function $u\left(  c\right)
$. They were first derived in \cite{EkL} for the special case $u\left(
c\right)  =\ln c$. Note that $\lambda\delta_{1}+(1-\lambda)\delta_{2}$ is the
arithmetic mean, and $\left(  \frac{\lambda}{\delta_{1}}+\frac{\left(
1-\lambda\right)  }{\delta_{2}}\right)  ^{-1}$ is the geometric mean.

\textbf{Example 2: }$u\left(  c,k\right)  =U\left(  c,k\right)  $

In that case, the criterion (\ref{biexponential.criterion}) becomes:%
\begin{equation}
\int_{0}^{\infty}\left(  \lambda e^{-\delta_{1}t}+\left(  1-\lambda\right)
e^{-\delta_{2}t}\right)  u\left(  c\left(  t\right)  ,k\left(  t\right)
\right)  dt \label{c1}%
\end{equation}

This problem was studied in \cite{EkL} for $u\left(  c,k\right)  =U\left(
c,k\right)  =\ln c$. In the case at hand, (\ref{c1}), we find:%
\begin{align*}
\overline{g}_{\lambda}(k)  &  =\lambda\delta_{1}+(1-\lambda)\delta_{2}%
-\frac{u_{2}^{\prime}\left(  f\left(  k\right)  ,k\right)  }{u_{1}^{\prime
}\left(  f\left(  k\right)  ,k\right)  }\\
\underline{g}_{\lambda}(k)  &  =\frac{1}{\frac{\lambda}{\delta_{1}}%
+\frac{\left(  1-\lambda\right)  }{\delta_{2}}}-\frac{u_{2}^{\prime}\left(
f\left(  k\right)  ,k\right)  }{u_{1}^{\prime}\left(  f\left(  k\right)
,k\right)  }%
\end{align*}

These are the same as the preceding ones, with the corrective term
$-u_{2}^{\prime}/u_{1}^{\prime}$.

\textbf{Example 3: }$u=u\left(  c\right)  \ $and $U=U\left(  c\right)  $

In that case, the criterion (\ref{biexponential.criterion}) becomes:%
\[
\lambda\int_{0}^{\infty}e^{-\delta_{1}t}u\left(  c\left(  t\right)  \right)
dt+\left(  1-\lambda\right)  \int_{0}^{\infty}e^{-\delta_{2}t}U\left(
c\left(  t\right)  \right)  dt
\]

We find:%
\begin{align*}
\overline{g}_{\lambda}(k)  &  =\frac{\lambda\delta_{1}{u_{1}^{\prime}}\left(
f\left(  k\right)  \right)  +(1-\lambda)\delta_{2}U_{1}^{\prime}\left(
f\left(  k\right)  \right)  }{\lambda u_{1}^{\prime}\left(  f\left(  k\right)
\right)  +(1-\lambda)U_{1}^{\prime}\left(  f\left(  k\right)  \right)  }\\
\underline{g}_{\lambda}(k)  &  =\frac{\lambda{u}_{1}^{\prime}\left(  f\left(
k\right)  \right)  +(1-\lambda)U_{1}^{\prime}\left(  f\left(  k\right)
\right)  }{\frac{\lambda}{\delta_{1}}u_{1}^{\prime}\left(  f\left(  k\right)
\right)  +\frac{\left(  1-\lambda\right)  }{\delta_{2}}U_{1}^{\prime}\left(
f\left(  k\right)  \right)  }%
\end{align*}

Here again we find $\underline{g}_{\lambda}(k)<\overline{g}_{\lambda}(k)$.
\newline\newline\noindent\textbf{We now proceed to the proof of Theorem
\ref{Th1}}. We apply Theorem \ref{find.the.strategies}, and we denote by
$V_{\lambda}$ and $W_{\lambda}$ the solution of equations (\ref{35}) and
(\ref{36}) with the boundary conditions (\ref{37}) and (\ref{38}). Set
$\sigma_{\lambda}\left(  k\right)  :=\varphi_{\lambda}\left(  V_{\lambda
}^{\prime}\left(  k\right)  ,k\right)  $, where $\varphi_{\lambda}$ is defined
by (\ref{39}) or (\ref{40}). We will now show that the corresponding
trajectory $k\left(  t\right)  $, defined by (\ref{equation.of.kb}), converges
to $k_{\infty}$. By Proposition \ref{equivalent.to.an.ODE.system}, this will
prove that $\sigma_{\lambda}$ is an equilibrium strategy.

We need the value of $V_{\lambda}^{\prime}(k_{\infty})$ and $\sigma_{\lambda
}^{\prime}(k_{\infty})$. Differentiating (\ref{35}), we find:
\begin{align}
V_{\lambda}^{\prime}(k_{\infty})f(k_{\infty})  &  -V_{\lambda}^{\prime
}(k_{\infty})\varphi_{\lambda}(V_{\lambda}^{\prime}(k_{\infty}),k_{\infty
})+u(\varphi_{\lambda}(V_{\lambda}^{\prime}(k_{\infty}),k_{\infty}),k_{\infty
})\nonumber\\
&  +\frac{1-\lambda}{\lambda}U(\varphi_{\lambda}(V_{\lambda}^{\prime
}(k_{\infty}),k_{\infty}),k_{\infty}) =u(f(k_{\infty}),k_{\infty}%
)+\frac{1-\lambda}{\lambda}U(f(k_{\infty}),k_{\infty}),\nonumber
\end{align}
thus $F(V_{\lambda}^{\prime}(k_{\infty}),k_{\infty})=0$ and hence:
\begin{equation}
V_{\lambda}^{\prime}(k_{\infty})=y^{\ast}(k_{\infty})=u_{1}^{\prime
}(f(k_{\infty}),k_{\infty})+\frac{1-\lambda}{\lambda}U_{1}^{\prime
}(f(k_{\infty}),k_{\infty}). \label{Vr'}%
\end{equation}

To compute $\sigma_{\lambda}^{\prime}(k_{\infty})$, we consider (\ref{c5}),
the integrated form of $V_{\lambda}^{\prime}(k)$ where we substitute
$h(t)=e^{-\delta t}$ and $H\left(  t\right)  =e^{-\delta_{2}t}$, and replace
$U$ with $\left(  1-\lambda\right)  \lambda^{-1}$. Evaluating (\ref{c5}) at
$k=k_{\infty}$, we get:%
\begin{align}
V_{\lambda}^{\prime}(k_{\infty})  &  =\frac{u_{1}^{\prime}(f(k_{\infty
}),k_{\infty})\sigma_{\lambda}^{\prime}(k_{\infty})+u_{2}^{\prime}%
(f(k_{\infty}),k_{\infty})}{\delta_{1}-f^{\prime}(k_{\infty})+\sigma_{\lambda
}^{\prime}(k_{\infty})}\nonumber\\
&  +\frac{1-\lambda}{\lambda}\frac{U_{1}^{\prime}(f(k_{\infty}),k_{\infty
})\sigma_{\lambda}^{\prime}(k_{\infty})+U_{2}^{\prime}(f(k_{\infty}%
),k_{\infty})}{\delta_{2}-f^{\prime}(k_{\infty})+\sigma_{\lambda}^{\prime
}(k_{\infty})}. \label{c6}%
\end{align}

Comparing (\ref{Vr'}) with (\ref{c6}), we can solve for $\sigma_{\lambda
}^{\prime}(k_{\infty})$. We get:
\begin{equation}
\sigma_{\lambda}^{\prime}(k_{\infty})=f^{\prime}(k_{\infty})-\delta_{1}%
\delta_{2}\frac{(\frac{1}{\delta_{1}}u_{1}^{\prime}(f(k_{\infty}),k_{\infty
})+\frac{1-\lambda}{\lambda\delta_{2}}\delta{U}_{1}^{\prime}(f(k_{\infty
}),k_{\infty}))(\underline{g}_{\lambda}(k_{\infty})-f^{\prime}(k_{\infty}%
))}{(u_{1}^{\prime}(f(k_{\infty}),k_{\infty})+\frac{1-\lambda}{\lambda}%
U_{1}^{\prime}(f(k_{\infty}),k_{\infty}))(\overline{g}_{\lambda}(k_{\infty
})-f^{\prime}(k_{\infty}))}, \label{c7}%
\end{equation}

Linearizing the equation of motion $\frac{dk}{dt}=f(k)-\sigma_{\lambda}(k)$ at
$k=k_{\infty}$ yields
\[
\frac{dk}{dt}=(f^{\prime}(k_{\infty})-\sigma_{\lambda}^{\prime}(k_{\infty
}))(k-k_{\infty}).
\]

Then $k$ converges to $k_{\infty}$ if $f^{\prime}(k_{\infty})-\sigma_{\lambda
}^{\prime}(k_{\infty})<0$. Comparing with (\ref{c7}), and writing in the
expressions of $\underline{g}_{\lambda}(k_{\infty})$ and $\overline
{g}_{\lambda}(k_{\infty})$ yields:%
\[
\delta_{1}\delta_{2}\frac{(\frac{1}{\delta_{1}}u_{1}^{\prime}(f(k_{\infty
}),k_{\infty})+\frac{1-\lambda}{\lambda\delta_{2}}{U}_{1}^{\prime}%
(f(k_{\infty}),k_{\infty}))(\underline{g}_{\lambda}(k_{\infty})-f^{\prime
}(k_{\infty}))}{(u_{1}^{\prime}(f(k_{\infty}),k_{\infty})+\frac{1-\lambda
}{\lambda}U_{1}^{\prime}(f(k_{\infty}),k_{\infty}))(\overline{g}_{\lambda
}(k_{\infty})-f^{\prime}(k_{\infty}))}<0.
\]
Because $u$ and $U$ are increasing with respect to their first variable, the
first factor in the numerator and denominator are positive. We are left with:%
\[
\frac{\underline{g}_{\lambda}(k_{\infty})-f^{\prime}(k_{\infty})}{\overline
{g}_{\lambda}(k_{\infty})-f^{\prime}(k_{\infty})}<0
\]
and the proof is complete.

\section{The Chichilnisky criterion}

In two influential papers \cite{Chi1}, \cite{Chi2}, Chichilnisky has proposed
an axiomatic approach to sustainable development, based on the twin ideas that
there should be no dictatorship of the present and no dictatorship of the
future. She suggests to use the following criterion
\begin{equation}
I_{\alpha}\left(  c\left(  \cdot\right)  ,k\left(  \cdot\right)  \right)
:=\delta\int_{0}^{\infty}u(c(t),k(t))e^{-\delta t}dt+\alpha\lim_{t\rightarrow
\infty}U(c(t),k(t)), \label{C-criterion}%
\end{equation}

The coefficient $\delta>0$ in front of the integral will make later formulas simpler.

\begin{lemma}
Suppose $u\left(  c,k\right)  \geq0$ and%
\[
\sup\left\{  U\left(  c,k\right)  \ |\ c>0,\ k>0,\ c=f\left(  k\right)
\right\}  =\infty.
\]
Then, for any $\alpha>0$,%
\[
\sup\left\{  I_{\alpha}\left(  c\left(  \cdot\right)  ,k\left(  \cdot\right)
\right)  \ |\ \left(  c\left(  \cdot\right)  ,k\left(  \cdot\right)  \right)
\in\mathcal{A}\left(  k_{0}\right)  \right\}  =\infty,
\]
where $\mathcal{A}(k_{0})$ is defined below Definition 1.
\end{lemma}

\begin{proof}
For fixed $k_{0}>0$, pick any $A>0$, and choose some constants $\left(
c_{1},k_{1}\right)  $ such that $c_{1}=f\left(  k_{1}\right)  $ and $U\left(
c_{1},k_{1}\right)  >A\alpha^{-1}$. Then choose some $c_{0}\in(0,c_{1})$ such
that $f\left(  k_{0}\right)  -c_{0}>0$. With every $T>0$, we associate the
path $\left(  c_{T}\left(  t\right)  ,k_{T}\left(  t\right)  \right)  $
defined by:%
\[
c_{T}\left(  t\right)  =\left\{
\begin{array}
[c]{l}%
c_{0}\text{ for }0\leq t\leq T,\\
c_{1}\text{ for }T\leq t,
\end{array}
\right.
\]
and denote the consumption-capital pair by $(c_{T},k_{T})$.

Since $f$ is increasing, we have $\frac{dk}{dt}\geq f\left(  k_{0}\right)
-c_{0}$ for all $t$, so eventually $k\left(  t\right)  $ will reach the value
$k_{1}$. Choose for $T$ the first time when $k\left(  t\right)  =k_{1}$.
Writing this into the criterion, and remembering $c_{1}=f\left(  k_{1}\right)
$, we find:%
\[
I_{\alpha}\left(  c_{T}\left(  \cdot\right)  ,k_{T}\left(  \cdot\right)
\right)  \geq\alpha\lim_{t\rightarrow\infty} U\left(  c_{T}(t),k_{T}%
(t)\right)  =\alpha U\left(  c_{1},k_{1}\right)  \geq A.
\]

\end{proof}

This result shows that very often it is not possible to optimize $I_{\alpha}$.
Even when it is, there is the time-inconsistency problem: successive
decision-makers will not agree on what the optimal solution is. This is seen
most easily by considering the following criterion:
\begin{equation}
\left(  1-\alpha\right)  \delta\int_{0}^{\infty}u(c(t),k(t))e^{-\delta
t}dt+\alpha{r}\int_{0}^{\infty}U(c(t),k(t))e^{-rt}dt. \label{E(r)-criterion}%
\end{equation}

When $r\rightarrow0$, the last term converges to $\alpha U\left(  c_{\infty
},k_{\infty}\right)  $, so that (\ref{E(r)-criterion}) converges to
(\ref{C-criterion}). On the other hand, Criterion (\ref{E(r)-criterion}) is a
special case of the biexponential criterion (\ref{biexponential.criterion})
with $\delta_{1}=\delta$, $\delta_{2}=r$ , and:
\begin{equation}
\frac{\alpha r}{\left(  1-\alpha\right)  \delta}=\frac{1-\lambda}{\lambda
},\ \lambda=\frac{\left(  1-\alpha\right)  \delta}{\alpha r+\left(
1-\alpha\right)  \delta},\ 1-\lambda=\frac{\alpha r}{\alpha r+\left(
1-\alpha\right)  \delta} \label{f}%
\end{equation}

So, for each $r>0$, the criterion\ (\ref{E(r)-criterion}) gives rise to a
time-inconsistent problem. Using the results in the preceding section, we find
a continuum of equilibrium strategies. Substituting (\ref{f}) into
(\ref{bar.g}) and (\ref{under.g}), we find the corresponding values:%

\begin{align*}
\overline{g}_{r}(k)  &  =\frac{\left(  1-\alpha\right)  \delta^{2}%
{u_{1}^{\prime}}\left(  f\left(  k\right)  ,k\right)  +\alpha r^{2}%
U_{1}^{\prime}\left(  f\left(  k\right)  ,k\right)  }{\left(  1-\alpha\right)
\delta u_{1}^{\prime}\left(  f\left(  k\right)  ,k\right)  +\alpha
rU_{1}^{\prime}\left(  f\left(  k\right)  ,k\right)  }-\frac{\left(
1-\alpha\right)  \delta u_{2}^{\prime}\left(  f\left(  k\right)  ,k\right)
+\alpha rU_{2}^{\prime}\left(  f\left(  k\right)  ,k\right)  }{\left(
1-\alpha\right)  \delta u_{1}^{\prime}\left(  f\left(  k\right)  ,k\right)
+\alpha rU_{1}^{\prime}\left(  f\left(  k\right)  ,k\right)  }\\
\underline{g}_{r}(k)  &  =\frac{\left(  1-\alpha\right)  \delta{u}%
_{1}^{\prime}\left(  f\left(  k\right)  ,k\right)  +\alpha rU_{1}^{\prime
}\left(  f\left(  k\right)  ,k\right)  }{\left(  1-\alpha\right)
u_{1}^{\prime}\left(  f\left(  k\right)  ,k\right)  +\alpha U_{1}^{\prime
}\left(  f\left(  k\right)  ,k\right)  }-\frac{\left(  1-\alpha\right)
u_{2}^{\prime}\left(  f\left(  k\right)  ,k\right)  +\alpha U_{2}^{\prime
}\left(  f\left(  k\right)  ,k\right)  }{\left(  1-\alpha\right)
u_{1}^{\prime}\left(  f\left(  k\right)  ,k\right)  +\alpha U_{1}^{\prime
}\left(  f\left(  k\right)  ,k\right)  }%
\end{align*}

The equations (\ref{35}), (\ref{36}) become:%
\begin{align}
\left(  f-\varphi_{r}\left(  V_{r}^{\prime}\right)  \right)  V_{r}^{\prime
}+u\left(  \varphi_{r}\right)  +\frac{\alpha r}{\left(  1-\alpha\right)
\delta}U\left(  V_{r}^{\prime}\right)  =  &  \frac{\delta+r}{2}V_{r}%
+\frac{\delta-r}{2}W_{r},\label{f1}\\
\left(  f-\varphi_{\lambda}\left(  V_{r}^{\prime}\right)  \right)
W_{r}^{\prime}+u\left(  \varphi_{\lambda}\right)  +\frac{\alpha r}{\left(
1-\alpha\right)  \delta}U\left(  V_{\lambda}^{\prime}\right)  =  &
\frac{\delta-r}{2}V_{r}+\frac{\delta+r}{2}W_{r} \label{f2}%
\end{align}
where $\varphi_{r}\left(  x,k\right)  $ is defined by:%
\[
u_{1}^{\prime}(\varphi_{r}(x,k),k)+\frac{\alpha r}{\left(  1-\alpha\right)
\delta}U_{1}^{\prime}\left(  \varphi_{r}(x,k),k\right)  =x
\]
and the boundary conditions (\ref{37}), (\ref{38}):%
\begin{align}
V_{r}\left(  k_{\infty}\right)   &  =\frac{1}{\delta}u\left(  f\left(
k_{\infty}\right)  ,k_{\infty}\right)  +\frac{\alpha}{\left(  1-\alpha\right)
\delta}U\left(  f\left(  k_{\infty}\right)  ,k_{\infty}\right)  ,\label{f3}\\
W_{r}\left(  k_{\infty}\right)   &  =\frac{1}{\delta}u\left(  f\left(
k_{\infty}\right)  ,k_{\infty}\right)  -\frac{\alpha}{\left(  1-\alpha\right)
\delta}U\left(  f\left(  k_{\infty}\right)  ,k_{\infty}\right)  . \label{f4}%
\end{align}

The equilibrium strategy is given by $\sigma_{r}\left(  k\right)  =\varphi
_{r}\left(  V_{r}^{\prime}\left(  k\right)  ,k\right)  $, and the
corresponding trajectory converges to $k_{\infty}$. Note that these strategies
are defined locally. More precisely, denote by $]k_{r}^{-},\ \ k_{r}^{+}[$ the
maximal interval of existence of the solution $\left(  V_{r},W_{r}\right)  $
of the ODE (\ref{f1}), (\ref{f2}) with the boundary condition (\ref{f3}),
(\ref{f4}), so that $k_{r}^{-}<k_{\infty}<\ k_{r}^{+}$.

We shall now solve the Chichilnisky problem by setting $r=0$ in the preceding
equations. We have:%
\begin{align}
\overline{g}_{0}(k)  &  =\delta-\frac{u_{2}^{\prime}\left(  f\left(  k\right)
,k\right)  }{u_{1}^{\prime}\left(  f\left(  k\right)  ,k\right)  }%
,\label{55}\\
\underline{g}_{0}(k)  &  =\frac{\left(  1-\alpha\right)  \delta{u}_{1}%
^{\prime}\left(  f\left(  k\right)  ,k\right)  }{\left(  1-\alpha\right)
u_{1}^{\prime}\left(  f\left(  k\right)  ,k\right)  +\alpha U_{1}^{\prime
}\left(  f\left(  k\right)  ,k\right)  }-\frac{\left(  1-\alpha\right)
u_{2}^{\prime}\left(  f\left(  k\right)  ,k\right)  +\alpha U_{2}^{\prime
}\left(  f\left(  k\right)  ,k\right)  }{\left(  1-\alpha\right)
u_{1}^{\prime}\left(  f\left(  k\right)  ,k\right)  +\alpha U_{1}^{\prime
}\left(  f\left(  k\right)  ,k\right)  } \label{55a}%
\end{align}

The pair $\left(  V_{0},W_{0}\right)  $ has to solve the following
boundary-value problem:%
\begin{align}
&  \left(  f\left(  k\right)  -\varphi_{0}\left(  V_{0}^{\prime},k\right)
\right)  V_{0}^{\prime}+u(\varphi_{0}\left(  V_{0}^{\prime},k\right)
,k)=\frac{\delta}{2}\left(  V_{0}+W_{0}\right)  ,\label{70}\\
&  \left(  f\left(  k\right)  -\varphi_{0}\left(  V_{0}^{\prime},k\right)
\right)  W_{0}^{\prime}+u(\varphi_{0}\left(  V_{0}^{\prime},k\right)
,k)=\frac{\delta}{2}\left(  V_{0}+W_{0}\right)  ,\label{71}\\
&  V_{0}\left(  k_{\infty}\right)  =\frac{1}{\delta}u\left(  f\left(
k_{\infty}\right)  ,k_{\infty}\right)  +\frac{\alpha}{\left(  1-\alpha\right)
\delta}U\left(  f\left(  k_{\infty}\right)  ,k_{\infty}\right)  ,\label{72}\\
&  W_{0}\left(  k_{\infty}\right)  =\frac{1}{\delta}u\left(  f\left(
k_{\infty}\right)  ,k_{\infty}\right)  -\frac{\alpha}{\left(  1-\alpha\right)
\delta}U\left(  f\left(  k_{\infty}\right)  ,k_{\infty}\right)  \label{73}%
\end{align}
with:%
\begin{equation}
u_{1}^{\prime}(\varphi_{0}(x,k),k)=x,\quad\varphi_{0}(u_{1}^{\prime
}(c,k),k)=c. \label{74}%
\end{equation}

Substracting (\ref{71}) from (\ref{70}), we get:%
\[
\left(  f\left(  k\right)  -\varphi_{0}\left(  V_{0}^{\prime},k\right)
\right)  \left(  V_{0}^{\prime}-W_{0}^{\prime}\right)  =0.
\]

Similarly to Lemma 9, it follows that $V_{0}^{\prime}-W_{0}^{\prime}=0$, and
so $V_{0}-W_{0}$ is a constant, namely:%
\[
W_{0}\left(  k\right)  =V_{0}\left(  k\right)  +W_{0}\left(  k_{\infty
}\right)  -V_{0}\left(  k_{\infty}\right)  =V_{0}\left(  k\right)
-2\frac{\alpha}{\left(  1-\alpha\right)  \delta}U\left(  f\left(  k_{\infty
}\right)  ,k_{\infty}\right)  .
\]

Writing this in (\ref{70}) and (\ref{71}), we find that $V_{0}\left(
k\right)  $ is a solution of the boundary-value problem:%
\begin{align*}
&  \left(  f\left(  k\right)  -\varphi_{0}\left(  V_{0}^{\prime},k\right)
\right)  V_{0}^{\prime}+u(\varphi_{0}\left(  V_{0}^{\prime},k\right)  ,k)
=\delta V_{0}-\frac{\alpha}{1-\alpha}U\left(  f\left(  k_{\infty}\right)
,k_{\infty}\right)  ,\\
&  V_{0}\left(  k_{\infty}\right)  =\frac{1}{\delta}u\left(  f\left(
k_{\infty}\right)  ,k_{\infty}\right)  +\frac{\alpha}{\left(  1-\alpha\right)
\delta}U\left(  f\left(  k_{\infty}\right)  ,k_{\infty}\right)  .
\end{align*}

Setting $V:=V_{0}-\frac{\alpha}{\left(  1-\alpha\right)  \delta}U\left(
f\left(  k_{\infty}\right)  ,k_{\infty}\right)  $, we see that $V$ is a
solution of the boundary-value problem:%
\begin{align}
\left(  f\left(  k\right)  -\varphi_{0}\left(  V^{\prime},k\right)  \right)
V^{\prime}+u(\varphi_{0}\left(  V^{\prime},k\right)  ,k)  &  =\delta
V,\label{81}\\
V\left(  k_{\infty}\right)   &  =\frac{1}{\delta}u\left(  f\left(  k_{\infty
}\right)  ,k_{\infty}\right)  . \label{82}%
\end{align}

This problem has been studied in \cite{Ek} (Section 2.4, case 2) when
$u\left(  c,k\right)  $ does not depend on $k$. In the general case, we have:

\begin{proposition}
\label{Pr1}If
\begin{equation}
f^{\prime}\left(  k_{\infty}\right)  \neq\delta-\frac{u_{2}^{\prime}\left(
f\left(  k_{\infty}\right)  ,k_{\infty}\right)  }{u_{1}^{\prime}\left(
f\left(  k_{\infty}\right)  ,k_{\infty}\right)  } \label{inequality.1}%
\end{equation}
the problem (\ref{81}), (\ref{82}) has two solutions $V_{1}$ and $V_{2}$,
defined on some non-empty half-interval $[k_{\infty},\ k_{\infty}+a)$ or
$(k_{\infty}-a,k_{\infty}]$. Both are $C^{1}$ on the half-interval, $C^{2}$ on
the interior, and have the same derivative at $k_{\infty}$, given by:%
\[
\varphi_{0}\left(  V_{i}^{\prime}\left(  k_{\infty}\right)  ,k_{\infty
}\right)  =c_{\infty}=f\left(  k_{\infty}\right)  ,\quad i=1,2.
\]

\end{proposition}

\begin{proof}
From now on, write $\varphi$ instead of $\varphi_{0}$. Rewrite (\ref{81}) as a
Pfaff system:
\begin{align}
dV  &  =pdk,\label{e.paff1}\\
p(f(k)-\varphi(p,k))+u(\varphi(p,k),k)  &  =\delta V \label{e.paff2}%
\end{align}

Differentiating (\ref{e.paff2}) leads to:
\begin{equation}
\delta dV=(f(k)-\varphi(p,k))dp+pf^{\prime}(k)dk+u_{2}^{\prime}(\varphi
(p,k),k)dk,\nonumber
\end{equation}
where we used (\ref{74}). Together with (\ref{e.paff1}), this yields
\begin{equation}
(f(k)-\varphi(p,k))dp=[\delta{p}-pf^{\prime}(k)-u_{2}^{\prime}(\varphi
(p,k),k)]dk. \label{e.system.tmp}%
\end{equation}

We have to investigate this system near the point $k=k_{\infty}$ and
$V=V_{\infty}$. Writing (\ref{e.paff2}) at this point, we get:%
\begin{equation}
p(f(k_{\infty})-\varphi(p,k_{\infty}))+u(\varphi(p,k_{\infty}),k_{\infty
})=\delta V_{\infty} \label{e1}%
\end{equation}
which has to be solved for $p$. Note that, because of (\ref{74}), we have:%
\begin{equation}
\min_{p}\{(f(k_{\infty})-\varphi(p,k_{\infty}))p+u(\varphi(p,k_{\infty
}),k_{\infty})\}=u(f(k_{\infty}),k_{\infty}) \label{e2}%
\end{equation}

In the case at hand, we have $\delta V_{\infty}=u\left(  f\left(  k_{\infty
}\right)  ,k_{\infty}\right)  $, but for the sake of completeness, and to have
a full description of the phase space in the $\left(  k,V\right)  $ plane, we
will first investigate the cases $\delta V_{\infty}<u(f(k_{\infty}),k_{\infty
})$ and $\delta V_{\infty}>u(f(k_{\infty}),k_{\infty})$.

\textbf{Case 0}: $\delta V_{\infty}<u(f(k_{\infty}),k_{\infty})$

Because of (\ref{e2}), equation (\ref{e1}) has no solution. So there are no
solutions going through $\left(  k_{\infty},V_{\infty}\right)  .$

\textbf{Case 1.} $\delta V_{\infty}>u(f(k_{\infty}),k_{\infty})$.

Equation (\ref{e1}) has two distinct solutions $p_{1}\neq p_{2}$. Note that
neither $p_{1}$ nor $p_{2}$ minimize the left hand side so $f(k_{\infty
})-\varphi(p_{i},k_{\infty})\neq0$ for $i=1,2$. We may therefore consider the
initial value problem:
\[
\frac{dp}{dk}=\frac{\delta{p}-pf^{\prime}(k)-u_{2}^{\prime}(\varphi
(p,k),k)}{f(k)-\varphi(p,k)},\quad p(k_{\infty})=p_{i}.
\]

It has a well-defined smooth solution $p_{i}(k)$, defined in a neighborhood of
$k_{\infty}$. We then define a function $V_{i}$ by:
\[
V_{i}(k):=\frac{1}{\delta}[(f(k)-\varphi(p,k))p+u(\varphi(p,k),k)].
\]

Notice that $V_{i}(k_{\infty})=V_{\infty}$, so $V_{i}(k)$ solves the initial
value problem (\ref{81}), (\ref{82}), with $V_{i}^{\prime}(k_{\infty})=p_{i}$.
Taking $i=1,2$, we find two solutions $V_{1}(k)$ and $V_{2}(k)$ of the same
initial value problem, with $V_{1}^{\prime}(k_{\infty})\neq V_{2}^{\prime
}\left(  k_{\infty}\right)  $.

\textbf{Case 2.} $\delta V_{\infty}=u(f(k_{\infty}),k_{\infty})$ and
$f^{\prime}(k_{\infty})\neq\delta-u_{2}^{\prime}/u_{1}^{\prime}$.

Equation (\ref{e1}) then has a single solution $p_{0}$, and we have
$f(k_{\infty})=\varphi(p_{0},k_{\infty})$, so that:
\[
p_{0}=u_{1}^{\prime}(\varphi(p_{0},k_{\infty}),k_{\infty})=u_{1}^{\prime
}(f(k_{\infty}),k_{\infty})>0.
\]

In this case, we shall use the same system (\ref{e.paff1})-(\ref{e.paff2}),
but we will take $p$ instead of $k$ as the independent variable. We consider
the initial value problem
\[
\frac{dk}{dp}=\frac{f(k)-\varphi(p,k)}{p(\delta-f^{\prime}(k)-\frac
{u_{2}^{\prime}(\varphi(p,k),k)}{p})},\quad k(p_{0})=k_{\infty}.
\]

It has a $C^{2}$ solution $k(p)$, defined in a neighborhood of $p=p_{0}$. We
associate with it a curve in the phase space $(k,V)$, defined in parametric
form by the equations
\begin{align*}
k  &  =k(p),\\
V  &  =\frac{1}{\delta}[(f(k)-\varphi(p,k))p+u(\varphi(p,k),k)].
\end{align*}

Then $V$ is also $C^{2}$ with respect to $p$ near $p=p_{0}$. Moreover, we
have
\begin{align*}
\frac{dV}{dp}  &  =\frac{1}{\delta}[(f(k)-\varphi(p,k)-\varphi_{1}^{\prime
}(p,k)p+u_{1}^{\prime}(\varphi(p,k),k)\varphi_{1}^{\prime}(p,k)]\\
&  +\frac{1}{\delta}[(f^{\prime}(k)-\varphi_{2}^{\prime}(p,k))p+u_{1}^{\prime
}(\varphi(p,k),k)\varphi_{2}^{\prime}(p,k)+u_{2}^{\prime}(\varphi
(p,k),k)]\frac{dk}{dp}\\
&  =\frac{f(k)-\varphi(p,k)}{\delta-f^{\prime}(k)-\frac{u_{2}^{\prime}%
(\varphi(p,k),k)}{p}}%
\end{align*}

Since $f(k_{\infty})=\varphi(p_{0},k_{\infty})$, we obtain
\begin{align*}
\frac{dk}{dp}(p_{0})  &  =0,\\
\frac{dV}{dp}(p_{0})  &  =0.
\end{align*}

This shows that the parametric curve $p\rightarrow\left(  k\left(  p\right)
,V\left(  p\right)  \right)  $ in the $\left(  k,V\right)  $-plane has a cusp
at $\left(  k_{\infty},V_{\infty}\right)  $. To find the type of the cusp, we
compute the second order derivatives with respect to $p$. We find:
\begin{align}
\frac{d^{2}k}{dp^{2}}(p_{0})  &  =\frac{-\varphi_{1}^{\prime}(p_{0},k_{\infty
})}{p_{0}[\delta-f^{\prime}(k_{\infty})-\frac{u_{2}^{\prime}(f(k_{\infty
}),k_{\infty})}{u_{1}^{\prime}(f(k_{\infty}),k_{\infty})}]}=\frac
{-u_{11}^{\prime\prime}(f(k_{\infty}),k_{\infty})}{u_{1}^{\prime}(f(k_{\infty
}),k_{\infty})[\delta-f^{\prime}(k_{\infty})-\frac{u_{2}^{\prime}(f(k_{\infty
}),k_{\infty})}{u_{1}^{\prime}(f(k_{\infty}),k_{\infty})}]}\neq0,\label{e3}\\
\frac{d^{2}V}{dp^{2}}(p_{0})  &  =\frac{-\varphi_{1}^{\prime}(p_{0},k_{\infty
})}{\delta-f^{\prime}(k_{\infty})-\frac{u_{2}^{\prime}(f(k_{\infty}%
),k_{\infty})}{u_{1}^{\prime}(f(k_{\infty}),k_{\infty})}}=\frac{-u_{11}%
^{\prime\prime}(f(k_{\infty}),k_{\infty})}{\delta-f^{\prime}(k_{\infty}%
)-\frac{u_{2}^{\prime}(f(k_{\infty}),k_{\infty})}{u_{1}^{\prime}(f(k_{\infty
}),k_{\infty})}}\neq0, \label{e4}%
\end{align}

This shows that the cusp has two branches, with a common tangent between them.
The branches extend on the left of $k_{\infty}$ if the left-hand side of
(\ref{e3}) is negative, that is, if:%
\begin{equation}
\delta-f^{\prime}(k_{\infty})-\frac{u_{2}^{\prime}(f(k_{\infty}),k_{\infty}%
)}{u_{1}^{\prime}(f(k_{\infty}),k_{\infty})}<0 \label{g}%
\end{equation}
and on the right if it is positive.The slope $m$ of the common tangent is
given by
\[
m=\frac{d^{2}V}{dp^{2}}(k_{0})/\frac{d^{2}k}{dp^{2}}(p_{0})=p_{0},
\]
and it is also the one-sided derivative of $V$ at $k_{\infty}$. Note for
future use that because $p_{0}=u_{1}^{\prime}(f(k_{\infty}),k_{\infty})$ we
have:
\[
\frac{d}{dk}\frac{1}{\delta}u(f(k),k)|_{k=k_{\infty}}=\frac{1}{\delta}\left(
u_{1}^{\prime}(f(k_{\infty}),k_{\infty})f^{\prime}(k_{\infty})+u_{2}^{\prime
}(f(k_{\infty}),k_{\infty}\right)  ,
\]
so that $m=p_{0}$ does not coincide with the tangent of the curve $\delta
V=u\left(  f\left(  k\right)  \right)  $ at $k=k_{\infty}$.

\textbf{Case 3}:\ $\delta V_{\infty}=u(f(k_{\infty}),k_{\infty})$ and
$f^{\prime}(k_{\infty})=\delta-u_{2}^{\prime}/u_{1}^{\prime}$

This is the case which was investigated in Theorem \ref{thm.1}. We have shown
that there exists a $C^{2}$ solution $V\left(  k\right)  $ on a neighbourhood
of $k_{\infty}\,.$

This concludes the proof of Proposition \ref{Pr1} (in fact, we only need Case 2).
\end{proof}

Let us summarize these results. The curve $\Gamma=\left\{  \left(  k,V\right)
\ |\ \delta V=u\left(  f\left(  k\right)  ,k\right)  \right\}  \,$\ separates
the plane in two regions.

\begin{itemize}
\item The region below the curve corresponds to Case 0:\ there are no solution there.

\item The region above the curve corresponds to Case 1: through each point
$\left(  k_{\infty},V_{\infty}\right)  $ there are two smooth solutions
intersecting transversally.

\item If $\left(  k_{\infty},V_{\infty}\right)  $ is on the curve, but
$f^{\prime}(k_{\infty})+\frac{u_{2}^{\prime}(f(k_{\infty}),k_{\infty})}%
{u_{1}^{\prime}(f(k_{\infty}),k_{\infty})}\neq\delta$, we are in Case 2. There
are two $C^{1}$ solutions defined only on one side of $k_{\infty}$. They are
tangent at $\left(  k_{\infty},V_{\infty}\right)  $, and transversal to
$\Gamma$.

\item If $\left(  k_{\infty},V_{\infty}\right)  $ is on the curve, and
$f^{\prime}(k_{\infty})+\frac{u_{2}^{\prime}(f(k_{\infty}),k_{\infty})}%
{u_{1}^{\prime}(f(k_{\infty}),k_{\infty})}=\delta$, we are in Case 3: there is
a $C^{2}$ solution defined on a neighbourhood of $k_{\infty}$.the
\end{itemize}

Figure 1 gives the phase diagram in the $\left(  k,V\right)  $ plane:

\begin{figure}[th]
\centering
\includegraphics[height=7.5cm]{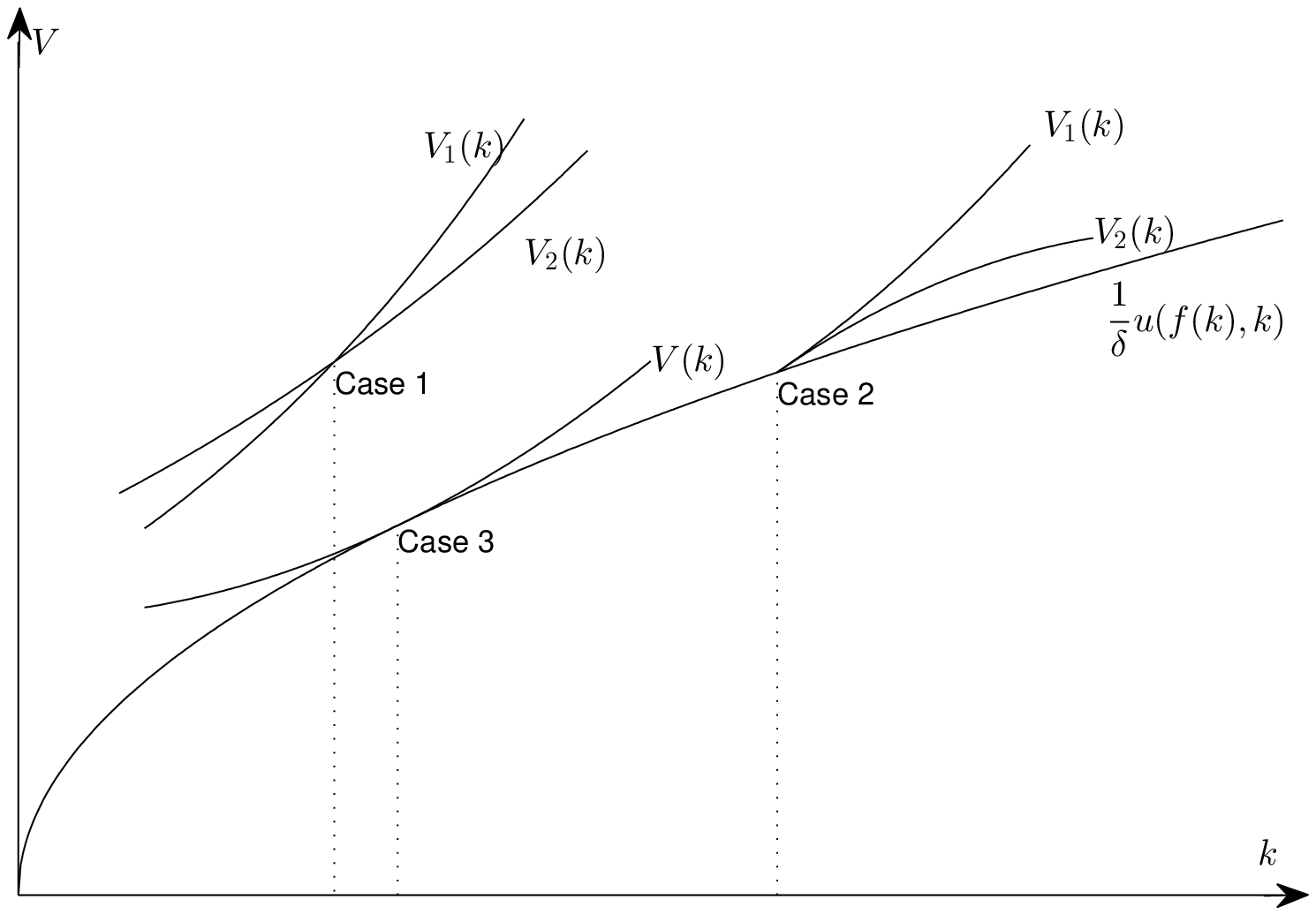} \center{Figure 1. The illustration of the solutions.}\end{figure}

Proposition \ref{Pr1} gives us two solutions, $V_{i},i=1,2$. Each of them
gives rise to an strategy $\sigma_{i}$ through the formula $\sigma_{i}\left(
k\right)  =\varphi_{0}\left(  V_{i}^{\prime}\left(  k\right)  ,k\right)  $
with $\varphi_{0}$ defined by (\ref{74}). The strategy $\sigma_{i}$ is $C^{0}$
on the half-interval, $C^{1}$ on its interior, with $\sigma_{i}\left(
k_{\infty}\right)  =f\left(  k_{\infty}\right)  $.

\begin{proposition}
One, and only one, of the strategies $\sigma_{1}$ and $\sigma_{2}$, converges
to $k_{\infty}$.
\end{proposition}

\begin{proof}
Consider the Euler-Lagrange equation (\ref{21}) for the Ramsey problem.\ The
phase diagram is given in Figure 2, where $\underline{k}$ is defined by
$f^{\prime}(\underline{k})=\delta-u_{2}^{\prime}/u_{1}^{\prime}$.
\begin{figure}[th]
\centering
\includegraphics[height=7.5cm]{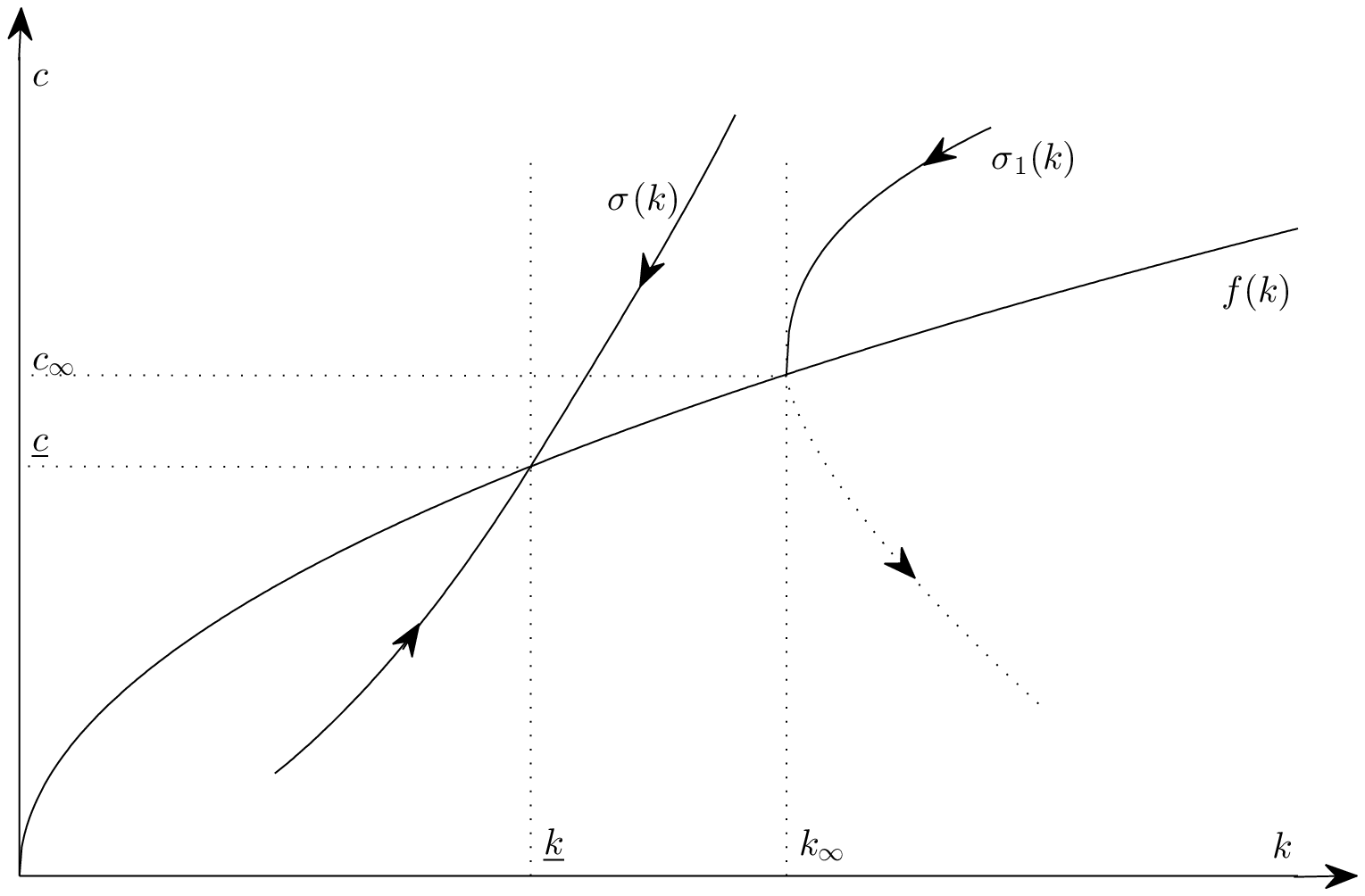} \center{Figure 2. The phase diagram for the Euler equation.}\end{figure}

If $k_{\infty}\neq\underline{k}$, there is one trajectory $\mathcal{T}$ going
through $\left(  k_{\infty},u\left(  f\left(  k_{\infty}\right)  ,k_{\infty
}\right)  \right)  $. The point $\left(  k_{\infty},u\left(  f\left(
k_{\infty}\right)  ,k_{\infty}\right)  \right)  $ separates it into two
branches, the upper one and the lower one.\ One of them goes to $\left(
k_{\infty},u\left(  f\left(  k_{\infty}\right)  ,k_{\infty}\right)  \right)
$, and the other one leaves it. These two branches are also the trajectories
associated with the two strategies $\sigma_{1}$ and $\sigma_{2}$, so one of
them converges and the other diverges.
\end{proof}

Since we are looking for a strategy which converges to $k_{\infty}$, we pick
the strategy $\sigma_{i}$, and the solution $V_{i}$, associated with the
branch oriented towards $k_{\infty}$. We will denote them by $\sigma$ and $V$.
We have proved the following:

\begin{theorem}\label{thm2}
Suppose $f^{\prime}\left(  k_{\infty}\right)  $ lies between $\overline{g}%
_{0}(k_{\infty})$ and $\underline{g}_{0}(k_{\infty})$. If $\underline{g}%
_{0}(k_{\infty})>\overline{g}_{0}(k_{\infty})$, then there exists an
equilibrium strategy $\sigma\left(  k\right)  $, defined on some interval
$]k_{\infty}-\kappa,\ k_{\infty}]\,$, which converges to $k_{\infty}$. It is
continuous on the interval, $C^{1}$ on its interior, with $\sigma\left(
k_{\infty}\right)  =f\left(  k_{\infty}\right)  $. If If $\underline{g}%
_{0}(k_{\infty})<\overline{g}_{0}(k_{\infty})$, there exists an equilibrium
strategy with the same properties, defined on some interval $[k_{\infty
},\ k_{\infty}+\kappa\lbrack$.
\end{theorem}

Note that one of the boundary values for $f^{\prime}\left(  k_{\infty}\right)
$, namely $\overline{g}_{0}(k_{\infty})$, corresponds to the solution of the
Ramsey problem ($\alpha=0$). Indeed, the equation $f^{\prime}\left(
k_{\infty}\right)  =\overline{g}_{0}(k_{\infty})$ coincides with equation
(\ref{51}).

We have thus identified a class of equilibrium strategies for the Chichilnisky
problem. They are one-sided, except when $f^{\prime}\left(  k_{\infty}\right)
=\overline{g}_{0}(k_{\infty})$, where we can apply Theorem \ref{thm2} to get a
strategy defined on a neighbourhood of $k_{\infty}$. For every other value of
$k_{\infty}$ satisfying (\ref{55}) and (\ref{55a}), the function $V(k)$ and
the strategy $\sigma\left(  k\right)  $ defined by $u_{1}^{\prime}\left(
\sigma\left(  k\right)  ,k\right)  =V^{\prime}\left(  k\right)  $ are defined
only on one of the two half-intervals limited by $k_{\infty}$. Suppose for
instance it is the right one, $[k_{\infty},\ k_{\infty}+\kappa\lbrack$. Then,
if $k_{\infty}\leq k_{0}<k_{\infty}+\kappa$, the equilibrium strategy will
bring $k_{0}$ to $k_{\infty}$ in finite time and stay there.

To our knowledge, this is the first time equilibrium strategies have been
found for the Chichilinisky criterion. Their economic interpretation, and
their detailed study, will be the subject of forthcoming work.

Let us give some examples.

\textbf{Example 1}: $u\left(  c\right)  =U\left(  c\right)  $

Neither depends on $k$, and we have:%
\begin{align*}
\overline{g}_{0}(k)  &  =\delta,\\
\underline{g}_{0}(k)  &  =\left(  1-\alpha\right)  \delta
\end{align*}

We have $\underline{g}_{0}(k)<\overline{g}_{0}(k)$, so the equilibrium
strategy exists only on the right hand side of $\underline{g}_{0}(k)$. The
existence condition is:%
\[
\left(  1-\alpha\right)  \delta<f^{\prime}(k_{\infty})<\delta
\]
and the equilibrium strategy $\sigma$ is defined on $[k_{\infty}%
,\ \infty\lbrack$. We denote $f^{\prime-1}(\delta)$ and $f^{\prime
-1}((1-\alpha)\delta)$ by $\underline{k}$ and $\overline{k}$ respectively.
There are three cases, depending on the position of the initial point $k_{0}$:

\begin{itemize}
\item If $k_{0}>\overline{k}$, then, for any $k_{\infty}\in]\underline{k}%
,\ \overline{k}[$, there exists an equilibrium strategy starting from $k_{0}$
and converging to $k_{\infty}$.

\item if $\underline{k}<k_{0}<\overline{k}$, then, for any $k_{\infty}%
\in]\underline{k},\ k_{0}[$, there exists an equilibrium strategy starting
from $k_{0}$ and converging to $k_{\infty}$.

\item if $k_{0}<\underline{k}$, the only equilibrium strategy starting from
$k_{0}$ is the optimal strategy for the Ramsey problem (that is, for the case
$\alpha=0$ ) which converges to the level $\underline{k}$ where $f^{\prime
}(\underline{k})=\delta$.
\end{itemize}

\textbf{Example 2}:\ $u\left(  c,k\right)  =U\left(  c,k\right)  $

We find:%
\begin{align*}
\overline{g}_{0}(k)  &  =\delta-\frac{u_{2}^{\prime}(f(k),k)}{u_{1}^{\prime
}\left(  f\left(  k\right)  ,k\right)  }\\
\underline{g}_{0}(k)  &  =\left(  1-\alpha\right)  \delta-\frac{u_{2}^{\prime
}(f(k),k)}{u_{1}^{\prime}\left(  f\left(  k\right)  ,k\right)  }%
\end{align*}

The existence condition is:%
\[
\left(  1-\alpha\right)  \delta-\frac{u_{2}^{\prime}(f(k),k)}{u_{1}^{\prime
}\left(  f\left(  k\right)  ,k\right)  }<f^{\prime}(k_{\infty})<\delta
-\frac{u_{2}^{\prime}(f(k),k)}{u_{1}^{\prime}\left(  f\left(  k\right)
,k\right)  }%
\]
and the equilibrium strategy $\sigma$ is defined on $[k_{\infty},\ k_{\infty
}+\kappa\lbrack$ for some $\kappa>0$. The situation is similar to the
preceding one, bearing in mind that now the strategy $\sigma$ may be defined
locally only

\textbf{Example 3:\ }$u=u\left(  c\right)  \,$\ and $U=U\left(  k\right)  $

In that case, we find:%

\begin{align*}
\overline{g}_{0}(k)  &  =\delta,\\
\underline{g}_{0}(k)  &  =\delta-\frac{\alpha}{\left(  1-\alpha\right)  }%
\frac{U_{2}^{\prime}\left(  f\left(  k\right)  ,k\right)  }{u_{1}^{\prime
}\left(  f\left(  k\right)  ,k\right)  }%
\end{align*}

There are two subcases:

\begin{itemize}
\item if $\frac{U_{2}^{\prime}\left(  f\left(  k\right)  ,k\right)  }
{u_{1}^{\prime}\left(  f\left(  k\right)  ,k\right)  }>0 $, then
$\underline{g}_{0}(k)<\overline{g}_{0}(k)$. The equilibrium strategy then is
defined on the right hand side of $k_{\infty}$, as in the preceding cases.

\item if $\frac{U_{2}^{\prime}\left(  f\left(  k\right)  ,k\right)  }
{u_{1}^{\prime}\left(  f\left(  k\right)  ,k\right)  }<0 $. The equilibrium
strategy then is defined on the left hand side of $k_{\infty}$.
\end{itemize}

Remember that $u_{1}^{\prime}>0$. If $U_{2}^{\prime}>0$, that is, if capital
is beneficial, we are in the first case. The second case corresponds to the
case when capital is detrimental, for instance if accumulating capital
accumulates pollution. In that case, the situation is inverted with respect to
the preceding examples:\ if the initial stock of capital is small enough, one
may stop capital accumulation at a level $k_{\infty}$ smaller than the level
$\underline{k}\ $where $f^{\prime}(\underline{k})=
\delta-\frac{\alpha}{\left(  1-\alpha\right)  }\frac{U_{2}^{\prime}\left(
f\left(  k\right)  ,k\right)  }{u_{1}^{\prime}\left(  f\left(  k\right)
,k\right)  }$. But there is no turning back: one can never achieve a
stationary level $k_{\infty}<k_{0}$. The situation is described in Figure 3

\begin{figure}[th]
\centering
\includegraphics[height=7.5cm]{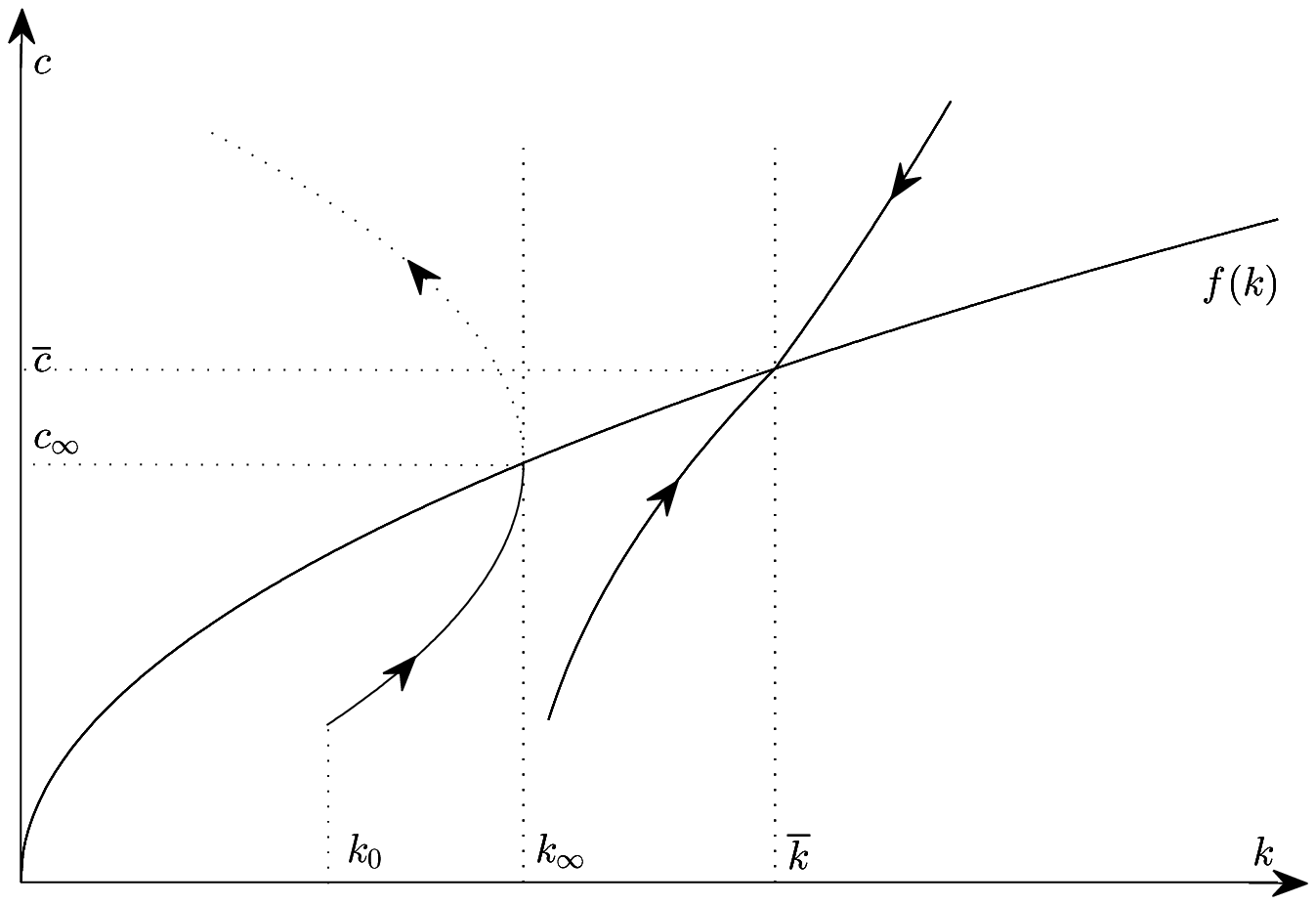} \center{Figure 3. The case when capital pollutes.}\end{figure}

\bigskip

\textbf{Acknowledgements.} The authors thank sincerely Professors Larry Karp
and Gerd Asheim for generously sharing their knowledge and ideas with them,
and pointing out several useful references. Qinglong Zhou and Yiming Long
would like to express their thanks to Professors Jean-Pierre Bourguignon, Jean
Dolbeault, Olivier Glass and Ivar Ekeland as well as the IHES and the CEREMADE
(Universit\'{e} Paris Dauphine) for supporting Qinglong's visits to IHES and
CEREMADE from February to April of 2012, and their warm hospitalities during
this time.

\end{document}